\documentclass[acmtocl]{acmsmall}
\pdfoutput=1

\usepackage{myhack}

\usepackage{ifthen}
\newboolean{commentsaon}         %
\setboolean{commentsaon}{true}
 \setboolean{commentsaon}{false}  %

\newboolean{commentson} %
\setboolean{commentson}{true}
 \setboolean{commentson}{false} %

\renewcommand{\comment}[1]
{\ifthenelse{\boolean{commentson}\AND\boolean{commentsaon}}
   {{\par\noindent\mbox{}{\small\blue[ *** #1 ]\par}\noindent\par}}{}}

\newcommand{\commenta}[1]
{\ifthenelse{\boolean{commentsaon}}
   {{\par\noindent\mbox{}{\small\color[rgb]{0, .5, 0}[ *** #1 ]\par}\noindent\par}}{}}
    \renewcommand{\firstfoot}{}
    \runningfoot{\mbox{ }}
    \def\copyrightline{}
    \def\doiline{}

\usepackage{xspace}
\usepackage{latexsym}
\usepackage{amssymb}   %
\usepackage{nameref}

\usepackage{color}
\newcommand\blue     {\color{blue}}

\usepackage{hyperref}
\usepackage[hyphenbreaks]{breakurl}
\usepackage{pgf}

\newcommand*{\seq}[2][n]
            {{\ensuremath{#2_{1}, \allowbreak \ldots, \allowbreak #2_{#1}}}}
\newcommand*{\SEQ}[3]
            {{\ensuremath{#1_{#2}, \allowbreak \ldots, \allowbreak #1_{#3}}}}
\newcommand*{\SEQC}[3]
            {{\ensuremath{#1_{#2} \cdots #1_{#3}}}}

\newcommand*{\notmodels}{\mathrel{\,\not\!\models}}
\newcommand*{\partto}{\hookrightarrow}

\newcommand*{\mydash}{{\mbox{\tt-}}}
\newcommand*{\HU}{{\ensuremath{\cal{H U}}}\xspace}
\newcommand*{\HB}{{\ensuremath{\cal{H B}}}\xspace}
\newcommand*{\TU}{{\ensuremath{\cal{T U}}}\xspace}
\newcommand*{\TB}{{\ensuremath{\cal{T B}}}\xspace}
\newcommand*{\M}{{\ensuremath{\cal M}}\xspace}

\usepackage[mathscr]{euscript} 
\renewcommand*{\S}{{\ensuremath{\mathscr S}}\xspace}

\newcommand*{\NN}{{\ensuremath{\mathbb{N}}}\xspace}
\newcommand*{\ZZ}{{\ensuremath{\mathbb{Z}}}\xspace}

\newcommand*{\mylonger}[1]{\makebox[0pt][l]{#1}}

\usepackage[ddmmyyyy]{datetime}

\begin{document}

\renewcommand{\today}{24-04-2015}
\markboth{\today}{\today}

\author{W{\l}odzimierz Drabent
  \affil{\\Institute of Computer Science,
         Polish Academy of Sciences,
         and
         IDA, Link\"oping University, Sweden
  }
}

\title
{Correctness and completeness of logic programs 
\hfill
\raisebox{0pt}[1ex][0ex]{\normalsize\small\sf%
  \begin{tabular}[t]{r@{}}
    Version of \today  \\
    \ifthenelse
             {\boolean{commentsaon}}
             {\normalsize\sf With comments}{}
  \end{tabular}%
}%
}



\terms{Theory, Verification}

\begin{abstract}
We discuss proving correctness and completeness of definite clause logic
programs.  We propose a method for proving completeness, 
while for proving correctness we employ a method which
should be well known but is often neglected.
Also, we show how to prove completeness and
correctness in the presence of SLD-tree pruning, and point out that
approximate specifications simplify specifications and proofs.

We compare the proof methods to declarative diagnosis (algorithmic
debugging), showing that approximate specifications eliminate a major
drawback of the latter.  We argue that our proof methods reflect natural
declarative thinking about programs, and that they can be used, formally or
informally, in every-day programming.

\end{abstract}

\maketitle

\commenta
{
ACM classes: 	D.1.6; F.3.1; D.2.4; D.2.5.
}

\comment{Previous abstract:

This paper discusses proving correctness and completeness of 
definite clause logic programs.
For proving program correctness
we use the method of Clark, 
which should be well known but is often neglected.
A contribution of this paper is a method for proving completeness.
We also show how to prove completeness and correctness 
in the presence of SLD-tree pruning.
We point out that approximate specifications simplify specifications and
proofs. 

  We compare the proof methods with declarative diagnosis (algorithmic
  debugging).  
  We show that approximate specifications remove a substantial drawback of
  declarative diagnosis.
  We argue that the proof methods correspond to natural declarative thinking
  about programs, and that they can be used, formally or informally, in
  every-day programming.
} %

\section{Introduction}
The purpose of this paper is to present proof methods to deal with 
the correctness related issues of definite clause logic programs.  The goal is
to devise methods which are
 1.~simple and applicable in the practice of Prolog programming, 
and 2.~declarative
(i.e.\ not referring to any operational semantics),
in other words depending only on the logical reading of programs.

The notion of program correctness (in imperative and functional programming) 
divides in logic programming into correctness and completeness.
Correctness means that all answers of the program are compatible with the specification,
completeness -- that the program  produces all the answers required by the
specification. 
A specification may be {\em approximate}\/:
 for such a specification
some answers are allowed, but not required to be computed.
We discuss various advantages of using approximate specifications.

There exists
a simple, natural, and moreover complete proof method for correctness
\cite {Clark79}.  It should be well known but has been mostly neglected.
Surprisingly little has been done on reasoning about completeness.
See Section \ref{sec:related} for the related work.
In this paper, for proving correctness we employ Clark's method.
For proving completeness we introduce a simplification of the method of
\cite{DBLP:journals/tplp/DrabentM05}.
We introduce a new notion of semi-completeness; semi-completeness and
termination imply completeness.
We discuss a few sufficient conditions for completeness.
We also propose a way of proving
that completeness is preserved under pruning of SLD-trees 
(due to e.g.\ the cut of Prolog).
This is augmented by a way of proving correctness for such trees,
taking into account that pruning may remove some answers.

We would like to treat logic programming as a declarative programming
paradigm, and to reason about programs declaratively, i.e.\ 
independently from their operational semantics.
The method of proving correctness 
is declarative, and so is the method of proving semi-completeness.
However some of the sufficient conditions for completeness are not declarative,
as they, roughly speaking,
 refer to an operational notion of program termination.
Notice that in most practical cases termination has to be established anyway.
So such a way of proving completeness may be a reasonable compromise
between declarative and non-declarative reasoning.

This paper provides small examples; for a more substantial application of the
presented proof methods see \cite{Drabent12.iclp}.
That paper presents a construction of a practical Prolog program, the SAT
solver of \cite{howe.king.tcs-shorter}.  Starting from a formal specification, 
a logic program is constructed hand in hand with its
correctness and completeness proofs.  Then control is added (delays, and
pruning SLD-trees) to obtain the final Prolog program.
The added control preserves correctness, while
the approach presented here makes it possible to formally prove that also
completeness is preserved.
This example illustrates also the usefulness of approximate specifications.
The consecutively constructed versions of the program are not equivalent --
their main predicates define different relations.
So it is not a case of
a program transformation which preserves the semantics of a program.
What is preserved is correctness and completeness w.r.t.\ 
the same approximate specification.

Closely related to proving correctness and completeness of programs
is program diagnosis, i.e.\ locating errors in programs.  We show a
correspondence between the presented proof methods and declarative diagnosis 
(called also algorithmic debugging).
We show how an important drawback of declarative diagnosis can be easily
overcome by employing approximate specifications.

\paragraph{Preliminaries}
In this paper we consider definite clause programs (i.e.\ logic programs without
negation).
For (standard) notation and definitions the reader is referred to
\cite{Apt-Prolog}. 
Given a predicate symbol $p$, by an {\em atom for} $p$ (or {\em $p$-atom})
 we mean an atom whose
predicate symbol is $p$, and by a {\em clause for} $p$ -- a clause whose head
is an atom for $p$.  
The set of the clauses for $p$ in the program under consideration is called
{\em procedure} $p$.

We assume a fixed alphabet of function and predicate symbols,
not restricted to the symbols from the considered program.
The only requirement is that the Herbrand universe is nonempty. In particular,
the set of function symbols may be finite or infinite.
The Herbrand universe will be denoted by \HU, the Herbrand base by \HB,
and the sets of all terms, respectively atoms, by \TU and \TB.
By $ground(E)$ we mean the set of ground instances
of a term, atom, etc $E$.
For a program $P$, $ground(P)$ is the set of ground instances of the
clauses of $P$.
The least Herbrand model of a program $P$ will be denoted by  $\M_P$.
It is the least fixed point of the immediate consequence operator
 $T_P\colon 2^\HB\to 2^\HB$; 
\
 $\M_P = \bigcup_{i=1}^\infty T_P^i(\emptyset)$.

  By a computed (respectively correct) answer for a program $P$ and a
  query $Q$ we mean an instance $Q\theta$ of $Q$ where $\theta$ is a computed
  (correct) answer substitution \cite{Apt-Prolog} for $Q$ and $P$.
  We usually say just
   {\em answer}
 as each computed answer is a correct one, 
  and each correct answer (for $Q$) is a computed answer
  (for $Q$ or for some its instance).
  Thus, by soundness and completeness of SLD-resolution,
  $Q\theta$ is an answer for $P$ iff $P\models Q\theta$.

By ``declarative'' (property, reasoning, \ldots) 
we mean referring only to logical reading of programs, 
thus abstracting from any operational semantics.
So $Q$ being an answer for $P$ is a declarative property.
Properties depending on the order of atoms in clauses will not be considered
declarative, as logical reading does not distinguish equivalent formulae,
like 
$\alpha\land\beta$ and $\beta\land\alpha$.

  Names of variables begin with an upper-case letter.
  We use the list notation of Prolog.  So 
  $[\seq t]$  ($n\geq0$) stands for the list of elements $\seq t$.
  Only a term of this form is considered a list.
 (Thus terms like $[a,a|X]$, or  $[a,a|a]$, where $a\neq [\,]$ is a constant,
  are not lists).
\NN stands for the set of natural numbers, and \ZZ for the set of integers;
 $f\colon A\partto B$  states that $f$ is a partial function from $A$ to $B$.
 We use a name ``proof tree'' instead of ``implication tree'' of
 \cite{Apt-Prolog}.

\paragraph{Organization of the paper}
The next section introduces the basic notions.
Section \ref{par:approximate} discusses approximate specifications.
Sections \ref{sec:correctness} and \ref{sec:completeness}
deal with proving, respectively, correctness and completeness of programs.
Proving completeness and correctness in the context of pruning of SLD-trees
is discussed in Section \ref{sec:pruning}.
Section \ref{sec:dd} discusses declarative diagnosis.  
Section \ref{sec:related} is an overview of related work,
Section \ref{sec:discussion} provides additional discussion of the proposed
methods, 
and the next section concludes the paper.  Two appendices include results
about completeness of the proof methods, and proofs
missing in the main text.

\pagebreak[3]

\section{Correctness and completeness}
\label{sec:corr+compl}
This section introduces basic notions of specification, correctness and
completeness, and relates them to the answers of programs.

The purpose of a logic program is to compute a relation, or a few relations.
It is convenient to assume that the relations are on the Herbrand universe.
A specification should describe them.  In most cases
all the relations defined by a program should be described,
an $n$-ary relation for each
$n$-argument predicate symbol. 
It is convenient to describe such family of relations as a Herbrand
interpretation $S$: 
a tuple $(\seq t)$ is in the relation corresponding to a predicate symbol $p$
iff \mbox{$p(\seq t)\in S$.}
So by a (formal) {\bf specification} we mean
 a Herbrand interpretation; given a specification $S$, each $A\in S$ is
called a {\em specified atom}.

In imperative programming, correctness usually means that the program results
are as specified.  In logic programming, due to its non-deterministic nature,
we actually have two issues: {\bf correctness} (all the results are
compatible with the specification) and {\bf completeness} (all the results
required by the specification are produced). 
In other words, correctness means that the relations defined by the program are
subsets of the specified ones, and completeness means inclusion in the
opposite  direction. 
In terms of specified atoms and the least Herbrand model $\M_P$ of a
program $P$, we define:
{\sloppy\par}
\begin{definition}
\label{def:corr:compl}
Let $P$ be a program and $S\subseteq\HB$ a specification.
$P$ is {\bf correct} w.r.t.\ $S$ when $\M_P\subseteq S$;
it is {\bf complete} w.r.t.\ $S$ when $\M_P\supseteq S$.
\end{definition}
We will sometimes skip the specification when it is clear from the context.
We propose to call a program {\bf fully correct} when it is both correct and complete.

It is important to relate the answers of programs
to their correctness and completeness.
The relation is not trivial as, for a program $P$ and a query $Q$,
 the equivalence 
$\M_P\models Q$ iff $P\models Q$ is not true in some cases.%
\footnote{
\label{footnote:example:notequivalent}
    For instance assume a two element alphabet of
    function symbols, with a unary $f$ and
    a constant $a$.  Take $P = \{\, p(a).\ p(f(Y)).\}$.  
    Then $\M_P= \HU$.  Hence $\M_P\models p(X)$ but $P\notmodels p(X)$.

It is quite often believed that the equivalence always holds,
See for instance erroneous
Th.\,2 in \cite{DBLP:conf/rweb/Paschke11},
or the remark on the least Herbrand models on p.\,512 of
\cite{Naish.TPLP06}.

}  

\pagebreak[3]

\begin{lemma}
\label{lemma:MP}
Let $P$ be a program, and $Q$ a query such that
\begin{enumerate}
\item 
\label{lemma:MP:condition}
 $Q$ contains exactly $k\geq0$ (distinct) variables, and
 the underlying language has (at least) $k$ constants not occurring in $P,Q$,
 or a non-constant function symbol not occurring in~$P,Q$.
\nopagebreak
\end{enumerate}
\nopagebreak
Then $\M_P\models Q$ iff $P\models Q$.
\end{lemma}
In particular, the equivalence holds for ground queries.
It also holds when the set of function symbols is infinite (and $P$ is finite).
Note that condition (\ref{lemma:MP:condition}) is satisfied by any Prolog
program and query, so the equivalence holds in practice.
See Appendix \ref{appendix} for a proof of the lemma,
and \cite{drabent.Herbrand.arxiv.2015} for a slightly more general result
and further discussion.

Now we are ready to relate program answers to the correctness and completeness 
of programs.
(Remember that a query $Q$ is an answer for $P$ iff $P\models Q$.)

\begin{proposition}
\label{prop:answers}
Let $P$ be a program, and $S$ a specification.

\begin{enumerate}
\item 
\label{prop:answers:cond1}
$P$ is correct w.r.t.\ $S$ \ iff \
 $P\models Q$ implies $S\models Q$ for each query $Q$.

\item
\label{prop:answers:cond2}
    $P$ is complete w.r.t.\ $S$ \ iff \ 
 $S\models Q$ implies $P\models Q$ for any {\em ground} query $Q$.

\item
    \label{prop:answers:cond:infinite}
If the set of function symbols of the underlying language is infinite 
and $P$ is finite then
\\
    $P$ is complete w.r.t.\ $S$ \ iff \
 $S\models Q$ implies $P\models Q$ for any query $Q$.

\item
     \label{prop:answers:cond3}
   $P$ is complete w.r.t.\ $S$ \ iff \ 
  \\\mbox{}\hfill
   $S\models Q$ implies $P\models Q$ for any query $Q$ such that
     condition (\ref{lemma:MP:condition}) of Lemma \ref{lemma:MP} holds.
\end{enumerate}
\enlargethispage*{3ex}
\nopagebreak
Within this proposition, ``query $Q$'' can be replaced by 
``atomic query $Q$''.
\end{proposition}
\pagebreak[3]

\begin{proof}
For the first case, assume $\M_P\subseteq S$ (i.e.\ $P$ is correct).
If $P\models Q$ then $\M_P\models Q$, hence $S\models Q$.
Conversely, 
assume that $P\models Q$ implies $S\models Q$ for any atomic query $Q$.
If $A\in\M_P$ then $P\models A$ (by Lemma \ref{lemma:MP}),
 hence $S\models A$.
Thus $\M_P\subseteq S$.

For cases  \ref{prop:answers:cond2},
\ref{prop:answers:cond:infinite}, \ref{prop:answers:cond3}
it is sufficient to show the ``if'' of case \ref{prop:answers:cond2}
and the ``only-if'' of case \ref{prop:answers:cond3}. 
For the ``only-if'' of case \ref{prop:answers:cond3}, 
assume   $S\subseteq \M_P$  ($P$ complete) and $S\models Q$
(for a $Q$ as described).
Then $\M_P\models Q$ and, by Lemma \ref{lemma:MP}, $P\models Q$.

For the ``if'' of case \ref{prop:answers:cond2}
assume the right hand side of the equivalence, for atomic queries.
So for every ground atom $A$ we have 
 that $S\models A$ implies $P\models A$, hence  $\M_P\models A$.
Thus   $S\subseteq \M_P$.
\end{proof}

The following auxiliary notions will be used.
\begin{definition}
\label{def:corr:compl:aux}
The set of $p$-atoms from a specification $S$ will be called 
the {\em specification for} $p$ given by $S$.

A {\em predicate} (or a procedure) $p$ in $P$ is
 {\em correct} w.r.t.\ $S$ when each $p$-atom of $\M_P$ is in $S$, and 
 {\em complete} w.r.t.\ $S$ when each $p$-atom of $S$ is in $\M_P$.

An {\em answer} $Q$ is  {\em correct} w.r.t.\ $S$ when $S\models Q$.

$P$ is {\em complete for a query} $Q$ w.r.t.\  $S$
when
$S\models Q\theta$ implies that $Q\theta$ is an answer for $P$,
for any ground instance $Q\theta$ of $Q$.
\end{definition}
Informally, complete for $Q$ means that all the answers for $Q$ required 
by the specification are computed.
Note that a program is complete w.r.t.\ S
 iff it is complete w.r.t.\ S for any query
 iff it is complete w.r.t.\ S for any query $A\in S$.

There is an alternative way to define program completeness, namely by
$S\models Q$ implying $P\models Q$, for every query $Q$
\cite{DBLP:journals/tplp/DrabentM05}.   
It was not chosen here,
as the two notions of completeness are equivalent in the practical case of an
infinite alphabet of function symbols.
Also, it seems that the chosen version leads to simpler sufficient conditions
for completeness.

\section{Approximate specifications}
\label{par:approximate}
This section first shows that requiring specifications to be exact 
leads to difficulties and inconveniences.  Then usefulness of approximate
specifications is discussed.

Notice that if a program $P$ is both correct and complete w.r.t.\ $S$
then $\M_P=S$ and the specification describes exactly the relations defined by
$P$.  Often it is difficult (and not necessary) to specify the relations
exactly.  
The relations defined by the program are often not exactly the ones intended by
the programmer.  Such discrepancy is however not an error,
because for all intended usages of the program the answers are the same as
for a program 
defining the intended relations.
For certain atoms it is irrelevant whether they are, or are not, in $\M_P$.
We illustrate this by some examples.

\begin{example}
\label{ex:spec-append}%
The well-known program
\[
{\rm APPEND} = \{\, \
  app(\,[H|K],L,[H|M]\,) \gets app(\,K,L,M\,). \ \ \ \
\linebreak[3]
  app(\,[\,],L,L\,). \
\,\}
\]  
does not define the relation of  list concatenation
(provided that the underlying alphabet contains some function symbol not
occurring in the program, e.g.\ a constant $a$).
For instance, ${\rm APPEND} \models app([\,],a,a)$.
In other words, APPEND is not correct w.r.t.\ 
\[
S_{\rm APPEND}^0 = \{\, app(k,l,m)\in\HB \mid
                k,l,m \mbox{ are lists, } k*l=m
\,\},
\] 
where $k*l$ stands for the concatenation of lists $k,l$.
It is however complete w.r.t.\ $S_{\rm APPEND}^0$, and correct w.r.t.\ 
\[
S_{\rm APPEND} = \{\, app(k,l,m)\in\HB \mid
                \mbox{if $l$ or $m$ is a list then }
                app(k,l,m)\in S_{\rm APPEND}^0
\,\}.
\] 
Proofs are given later on in
 Examples \ref{ex:append-corr}, \ref{ex:append-compl}.
Correctness w.r.t.\,$S_{\rm APPEND}$ and completeness 
w.r.t.\,$S_{\rm APPEND}^0$ are sufficient to show that
APPEND will produce the required 
results when used to concatenate or split lists.%
\footnote{
 Formally:
 In such usage, in a query $Q=app(s,t,u)$  the
 argument $t$ (and $s$) is a list, or $u$ is a list.
 By the correctness w.r.t.\,$S_{\rm APPEND}$
 and Prop.\,\ref{prop:answers}(\ref{prop:answers:cond1}),
 for any answer $Q\theta$ of APPEND we have  $S_{\rm APPEND} \models Q\theta$.  
 As $t\theta$ or $u\theta$ is a list, each ground instance of $Q\theta$ is in 
  $S_{\rm APPEND}^0$.  Hence  $S_{\rm APPEND}^0\models Q\theta$.
  This means that $s\theta, t\theta, u\theta$ are lists and $u\theta$ is the
  concatenation of  $s\theta, t\theta$.

 Conversely, 
 let us use the fact that APPEND is complete w.r.t.\,$S_{\rm APPEND}^0$
 for any set of function symbols that includes $[\,]$ and $[\ |\ ]$. 
  Whenever $u\sigma$  is a ground list which is the concatenation of 
  $s\sigma$ and $t\sigma$ (equivalently $S_{\rm APPEND}^0 \models Q\sigma$)
  then, by Prop.\,\ref{prop:answers}(\ref{prop:answers:cond2}), 
  $Q\sigma$ is an answer of APPEND.
  The same holds for a non-ground list $u\sigma$, provided that the condition of
  Prop.\,\ref{prop:answers}(\ref{prop:answers:cond:infinite})
  or \ref{prop:answers}(\ref{prop:answers:cond3}) holds.
  In particular, when $u=[\seq u]$ is a list and $s,t$ are distinct variables
  then each  $app( [\seq[i]u],[\SEQ u {i+1} n ],u)$ ($i=0,\ldots,n$)
  is an answer of APPEND.
  In other words, the program produces each splitting of $u$.

When the conditions of Prop.\,\ref{prop:answers} are not satisfied, it is
possible that a program $P$ is 
complete w.r.t.\,$S_{\rm APPEND}^0$ and $S_{\rm APPEND}^0\models Q$, but
$Q$ is not an answer of $P$. 
This happens for instance when
$[\,]$ and $[\ |\ ]$ are the only function symbols,
 $Q=app([X],[\,],[X])$, and
$P$ consists of three clauses:
  $  app(\,[\,],L,L\,)$.;
  $app(\,[[\,]|K],L,[[\,]|M]\,) \gets app(\,K,L,M\,)$.;
  $app(\,[[H|T]|K],L,[[H|T]|M]\,) \gets app(\,K,L,M\,)$.
(A completeness proof for $P$ is similar to that of Ex.\,\ref{ex:append-compl}).

} 
 Understanding the exact relation defined by the program is unnecessary.
 See \cite{DBLP:journals/tplp/DrabentM05} for further discussion.

Note that $S_{\rm APPEND}$ is
the union of $S_{\rm APPEND}^0$ 
and the complement  $\HB\setminus T$ of the set
$
T = \{\, app(k,l,m)\in\HB \mid
                \mbox{ $l$ or $m$ is a list}
\,\}.
$
The intuition is that it is irrelevant whether the atoms from $\HB\setminus T$
 are in $\M_{\rm APPEND}$. 

When we are only interested in using APPEND to concatenate lists,
the specification for correctness can be weakened by replacing
``if $l$ or $m$ is a list''  by ``if $l$ is a list''.
Similarly, if only list splitting is of interest then ``if $m$ is a list''
can be used. 
\end{example}

A similar example is given in Sect.\,\ref{sec:dd-drawback}
where in an insertion sort program it is irrelevant how elements should be
inserted into an unsorted list. 

\begin{example}
Programmers often do not know the exact relations defined by logic programs.
For many programs in textbooks, the exact relations defined by the
predicates are not explained.  For instance this concerns the programs in
Chapter 3.2
of \cite{Sterling-Shapiro-short}
defining predicates
{\tt member/2}, {\tt sub{}list/2}, {\tt prefix/2}, {\tt suffix/2},
and {\tt append/3}.
Similarly to the case of APPEND described above, the relations
defined by the programs are not what we understand as the list membership
relation, the ${\it sub list}$ relation etc.  
The programs define certain supersets of these relations.
(They are complete w.r.t.\ specifications describing the relations.)
Suitable specifications for correctness can be easily
constructed, similarly as for APPEND.  Understanding the exact relations
defined by the programs is unnecessary, and they are not discussed in 
the book.

The exact relations defined by programs are quite often misunderstood.
\noindent
For instance, in
\cite[Ex.\,15]{DevilleLau94}
it is claimed that the constructed program, let us call it $P$,
defines the relation of list inclusion
\[
{\cal I}({\it included}) =
\left\{\, {\it included}(l_1,l_2)\in\HB\ \left| \
\begin{array}{l}
     l_1,l_2 \mbox{ are lists},
\\
\mbox{every element of $l_1$ belongs to $l_2$}  
\end{array}
\right.\right\}.
\]
In our terms, this means that predicate ${\it included}$ of $P$
 is correct and complete w.r.t.\ ${\cal I}({\it included})$.
However the correctness does not hold;
$P$ contains a clause ${\it included}(L_1,L_2)\gets L_1=[\,].$,
thus $P\models{\it included}([\,],t)$ for any term $t$.

\end{example}

The examples show that it is often not necessary to know the semantics of a
program exactly, and that its approximate description is sufficient.
An {\bf approximate specification} is a pair of specifications
$S_{\it c o m p l}, S_{corr}$ called, respectively,
{\em specification for completeness} and {\em specification for correctness}.
The intention is that the program is complete w.r.t.\ the former, and correct
w.r.t.\ the latter.
We say that a program $P$ is {\em fully correct} w.r.t.\ 
$S_{\it c o m p l}, S_{corr}$ when
$S_{\it c o m p l}\subseteq \M_P \subseteq S_{corr}$.
In other words, the specifications $S_{\it c o m p l}, S_{corr}$ describe,
 respectively, which atoms have to be computed, and
which are allowed to be computed;
for the atoms from $S_{corr} \setminus  S_{\it c o m p l}$
the semantics of the program is irrelevant.
Sometimes, by abuse of terminology, a single specification $S_{corr}$ or  $S_{\it c o m p l}$ will be
called approximate, especially when the intention is that that the
specification is distinct from $\M_P$.

Usually programs are intended to be used with queries $Q$ such that
if $Q$ is atomic then
it has no instances in $S_{corr}\setminus S_{\it c o m p l}$.
For such queries a program fully correct w.r.t.\ $S_{\it c o m p l}, S_{corr}$ 
behaves like a program correct and complete w.r.t.\ $S_{\it c o m p l}$:

\begin{proposition}
Let $S_{\it c o m p l}, S_{corr}$ be an approximate specification,
$P$ be a program complete w.r.t.\ $S_{\it c o m p l}$ and correct
w.r.t.\  $S_{corr}$, and $P'$ a program correct and complete w.r.t.\ 
$S_{\it c o m p l}$.
Let $Q$ be an atomic query such that 
$ground(Q)\cap (S_{corr}\setminus S_{\it c o m p l}) = \emptyset$.
Let $Q\theta$ be an instance of $Q$ such that 
condition (\ref{lemma:MP:condition}) of Lemma \ref{lemma:MP} holds 
for $P$ with $Q\theta$ and for $P'$ with $Q\theta$.
Then $Q\theta$ is an answer for $P$ iff  $Q\theta$ is an answer for $P'$.
\end{proposition}
In particular, the equivalence holds when the set of function symbols 
is infinite (and $P,P'$ are finite).

\begin{proof}
For
any program $P''$ such that $P''$ and $Q$ satisfy (\ref{lemma:MP:condition}),
it holds that
$P''\models Q\theta$ iff $ground(Q\theta)\subseteq M_{P''}$.
Note that   $\M_{P'}=S_{\it c o m p l}\subseteq \M_P \subseteq  S_{corr}$.
We have
$ground(Q)\cap S_{corr}   \subseteq S_{\it c o m p l}$.
Hence  $ground(Q\theta)\cap \M_P  \subseteq S_{\it c o m p l}$.
Thus 
$ground(Q\theta)\subseteq \M_P$ iff
$ground(Q\theta)\subseteq S_{\it c o m p l}$.
Hence $P\models Q\theta$ iff  $P'\models Q\theta$.
\end{proof}

 \vspace{0ex}

\paragraph{Comments}
\label{par:approximate:comments}
An obvious role of approximate specifications is to describe
selected properties of programs.
  Ex.\,\ref{ex:split-corr} below shows a case where
only the lengths of lists are taken into account, not their elements.
Note that such a property cannot be expressed by an exact specification. 

Typed logic programming \cite{types.lp.92-short} is outside of the scope of this
paper.  But it is worth mentioning here that, 
for some programs, introducing types may change their semantics,
so that they define exactly the intended relations.
For instance in some appropriate typed logic,
the APPEND program of Ex.\,\ref{ex:spec-append}
(with the arguments of $app$  typed to be lists)
would define exactly the relation of list concatenation
Thus such typed APPEND would be correct and complete w.r.t.\ $S_{\rm APPEND}^0$,
and the need for an approximate specification disappears.
On the other hand, 
to obtain a similar result for procedure ${\it insert}$ of 
Ex.\,\ref{ex:inc-diag2} it is necessary to employ a type of sorted integer
lists, 
which seems to be outside of the scope of known type systems for logic programs.

Often, while developing a program, the exact relations to be defined 
should not be fixed in advance.  This would mean taking some design decisions
too early. 
For instance, constructing an append program with
$S_{\rm APPEND}^0$ (or with $S_{\rm APPEND}$) as an exact 
specification would lead to a program less efficient and
more complicated than APPEND.
See also Ex.\,\ref{ex:inc-diag2}.
A larger example is shown in
 [\citeNP{Drabent12.iclp};\,\citeyearNP{drabent.corr.2012}]:
a program is derived by stepwise refinement,
and the semantics of (the common predicates in) the consecutive versions of the
program differ.  
What does not change is an approximate specification (of the common predicates);
the correctness and completeness is preserved in all the versions.

This example suggests a generalization of the paradigm of program development
by transformations which preserve program semantics 
 \cite[and the references therein]{PettorossiPS10shorter}.
It shows that it is useful and natural to use transformations which instead
preserve correctness and completeness w.r.t.\ an approximate specification.

Importance of approximate specifications for declarative diagnosis is 
discussed in Sect.\,\ref{sec:dd}.

\section{Proving correctness}
\label{sec:correctness}
To prove correctness we use the following property
\cite {Clark79}.
\begin{theorem}[Correctness]
\label{th:correctness}
A sufficient condition for a program $P$ to be correct w.r.t.\ 
a specification $S$ is
\nopagebreak
\vspace{-2ex}
\[
S\models P.
\]
\end{theorem}
So the sufficient condition states that,
     for each ground instance 
    $
    H\gets \seq B
    $
    ($m\geq0$)
    of a clause of the program,
     if $\seq B\in S$ then $H\in S$.
Equivalently,  $T_P(S)\subseteq S$.

\begin{proof}
$S\models P$ means that $S$ is a Herbrand model of $P$, thus $\M_P\subseteq S$.
\end{proof}

\begin{example}
\label{ex:split-corr}%
Consider a program SPLIT:
\begin{eqnarray}
  &&     s( [\,], [\,], [\,] ).     \label{split1}\\   
  &&     s( [X|X s], [X|Y s], Zs ) \gets s( X s, Z s, Y s ).   \label{split2}
\end{eqnarray}
\newcommand*{\specsplit}{\ensuremath{S}\xspace}
Let $|l|$ denote the length of a list $l$.
A specification
\[
\specsplit = \{\, s(l, l_1, l_2) \mid 
                 l, l_1, l_2 \mbox{ are lists, } 0\leq |l_1|-|l_2| \leq 1 \,\}
\]
describes 
how the sizes of the last two arguments of $S$ are related.
To show, by Th.\,\ref{th:correctness}, 
that the program is correct w.r.t.\ \specsplit,
consider a ground instance 
$s( [h|t], [h|t_2], t_1)\gets s( t, t_1, t_2)$ of (\ref{split2}).
Assume $s(t, t_1, t_2)\in\specsplit$. Thus $ [h|t], [h|t_2], t_1$ are lists.
Let $m= |t_1|-|t_2|$.  As $m\in\{0,1\}$, we have
$|[h|t_2]|-|t_1| = 1- m \in\{0,1\}$.  So the head 
$s( [h|t], [h|t_2], t_1)$ is in \specsplit.
The proof for  (\ref{split1}) is trivial.

A stronger specification with respect to which SPLIT is correct is shown 
below.
\end{example}
\begin{example}
\label{ex:append-corr}
To show that program APPEND is correct w.r.t.\  $S_{\rm APPEND}$ from 
Ex.\,\ref{ex:spec-append}, 
take a $(H\gets B)\in ground({\rm APPEND})$ such that
$B = app(k,l,m)\in S_{\rm APPEND}$.
So $H = app(\,[h|k],l,[h|m]\,)$.
If $l$ or $[h|m]$ is a list then 
$k,l,m$ are lists and $k*l=m$.  So
$[h|k],l,[h|m]$ are lists and
$[h|k]*l=[h|m]$.  Thus $H\in S_{\rm APPEND}$.

For the other clause of the program, take
$app([\,], m,m)\in ground({\rm APPEND})$.
If $m$ is a list then $[\,]*m=m$; so 
$app([\,], m,m)\in S_{\rm APPEND})$.

\end{example}

\begin{example}
\label{ex:split-corr2}%
Now we show correctness of SPLIT w.r.t.\ a more precise specification
\[
S_{\rm SPLIT} = 
\left\{\,
\begin{array}{l}
 s([\seq[2n]t], [t_1,\cdots,t_{2n-1}], [t_2,\cdots,t_{2n}] ),
\\
 s([\seq[2n+1]t], [t_1,\cdots,t_{2n+1}], [t_2,\cdots,t_{2n}] )
\end{array}
   \,
   \left|
   \,
   \begin{array}l
      n\geq0,\\ \seq[2n+1] t\in \HU  \\
   \end{array}
\right.\right\}
,
\]
where
$[t_k,\cdots,t_l]$ denotes the list $[t_k,t_{k+2},\ldots,t_l]$,
for $k,l$ being both odd or both even.
(When $k>l$, $[t_k,\cdots,t_l]=[t_k,\ldots,t_l]=[\,]$.)
For  (\ref{split1}) the proof is trivial.
Consider a ground instance $H\gets B$ of (\ref{split2}).
Assume $B\in S_{\rm SPLIT}$.  
If
$B= s([\seq[2n]t], \linebreak[3]
      [t_1,\cdots,t_{2n-1}],\linebreak[3] [t_2,\cdots,t_{2n}] )$
then
$H=s([t_0,\seq[2n]t], \linebreak[3]
     [t_0,t_2,\cdots,t_{2n}], [t_1,\cdots,t_{2n-1}] )$,
for some $t_0$.
If $B= s([\seq[2n+1]t],\linebreak[3]
         [t_1,\cdots,t_{2n+1}],\linebreak[3] [t_2,\cdots,t_{2n}] )$
then
$H= s([t_0,\seq[2n+1]t], [t_0,t_2,\cdots,t_{2n}], [t_1\cdots,t_{2n+1}] )$.
In both cases $H\in S_{\rm SPLIT}$.
(To see this, rename each $t_i$ as $t'_{i+1}$.)
\end{example}

Completeness of this correctness proving method is discussed in 
Appendix
\ref{appendix:compl}. 
For more examples
(among others involving difference lists and accumulators),
and for further explanations, references and discussion
see \cite {DBLP:journals/tplp/DrabentM05}.

\section{Proving completeness}
\label{sec:completeness}
We first introduce a notion of semi-completeness, and sufficient conditions
under which semi-completeness of a program implies its completeness.
The next subsection presents a sufficient condition for semi-completeness.
For cases where completeness does not follow from semi-completeness, we
present a way of proving completeness directly.

\subsection{Semi-completeness}
Before discussing semi-completeness we present a few auxiliary notions,
which have been introduced in a context of proving program termination.

\begin{definition}
A {\bf level mapping} is a
function $|\ |\colon \HB\to \NN$ assigning natural numbers to atoms.

A program $P$ is  {\bf recurrent} {\em w.r.t.\ a level mapping}~$|\ |$
\cite{DBLP:journals/jlp/Bezem93,Apt-Prolog} if, in
every ground instance  $H\gets\seq B\in ground(P)$ of its clause ($n\geq0$),
$|H|>|B_i|$ for all $i=1,\ldots,n$.
A program is {\em recurrent}
if it is recurrent w.r.t.\ some level mapping.   
\end{definition}

\begin{definition}
\label{def:acceptable}%
A program $P$ is {\bf acceptable} w.r.t.\ a specification $S$ and a level
mapping $|\ |$ if 
$P$ is correct w.r.t.\ $S$, and for every
$H\gets\seq B\in ground(P)$
we have $|H|>|B_i|$ whenever $S\models B_1,\ldots,B_{i-1}$.
A program is {\em acceptable} if it is acceptable w.r.t.\ some level mapping
and some specification.

\end{definition}

This definition is more general than the original one
\cite{AP93,Apt-Prolog}
which uses a model of $P$, instead of a specification for which $P$ is
correct. 
Both definitions are equivalent.%
\footnote{
  They define the same programs to be acceptable.
  More precisely,
  a program $P$ is acceptable w.r.t.\ 
  $|\ |$ and a specification $S$ 
  iff, according to the original definition,
  $P$ is acceptable w.r.t.\ $|\ |$ and a model $I$ of $P$.
  (The ``only-if'' case holds for each model $I$ such that 
  $\M_P\subseteq I\subseteq S$,
  the ``if'' case for each $S$ such that  $\M_P\subseteq S\subseteq I$.)
}  

We now introduce a notion of semi-completeness.  It may be seen as a step on
the way to proving completeness.

\begin{definition}
A program $P$ is {\bf semi-complete} 
w.r.t.\ a specification $S$ if
$P$ is complete w.r.t.\ $S$ for 
(cf.\ Df.\,\ref{def:corr:compl:aux})
any query $Q$ for which there exists a finite SLD-tree.
\end{definition}

In practice, the existence of a finite SLD-tree means
that $P$ with $Q$ terminates under some selection rule.
For a semi-complete program, if a computation for a query $Q$
terminates then
it has produced all the answers for $Q$ required by the specification.
Obviously, a complete program is semi-complete.
We also have:

\begin{proposition}[Completeness]
\label{prop:semi-compl}%
Let a program $P$ be semi-complete w.r.t.\ $S$. 
The program is complete w.r.t\ $S$ if
\begin{enumerate}
\item 
\label{prop:semi-compl:term}
for each query $A\in S$ there exists a finite SLD-tree, or
\item 
\label{prop:semi-compl:recu}
the program is recurrent, or
\item 
\label{prop:semi-compl:accept}
the program is acceptable
(w.r.t.\ a specification $S'$ possibly distinct from $S$).
\end{enumerate}
\end{proposition}
Note that condition (\ref{prop:semi-compl:term}) is equivalent 
to each $A\in S$ being an instance of a query $Q$ with a finite SLD-tree.

\begin{proof}
For a program $P$ semi-complete w.r.t.\ $S$, the first part of condition 
(\ref{prop:semi-compl:term})  implies that 
$P$ is complete w.r.t.\ $S$ for each query $A\in S$;
hence $S\subseteq \M_P$.
Condition (\ref{prop:semi-compl:recu}), or (\ref{prop:semi-compl:accept})
implies that for any ground query
each SLD-tree, respectively each LD-tree,
is finite \cite{Apt-Prolog}.  
This implies (\ref{prop:semi-compl:term}).
\end{proof}

\subsection{Proving semi-completeness}
\label{sec:semi-compl}
We need the following notion.
\begin{definition}
  A ground atom $H$ is
  {\bf covered} {\em by a clause} $C$ w.r.t.\ a specification $S$
  if $H$ is the head of a ground instance  
  $
  H\gets \seq B
  $
  ($n\geq0$) of $C$, such that all the atoms $\seq B$ are in $S$.

  A ground atom $H$ is {\em covered by a program} $P$ w.r.t.\ $S$
  if it is covered w.r.t.\ $S$ by some clause $C\in P$.
  A specification $S$ is {\em covered by a program} $P$
  if each $H\in S$ is covered by $P$ w.r.t.\ $S$.

\end{definition}
For instance, given a specification 
$S = \{\, p(s^i(0))\mid i\geq0 \,\}$,
atom $p(s(0))$ is covered both by
$p(s(X))\gets p(X)$ and by
$p(X)\gets p(s(X))$.
The notion of an atom covered by a clause stems from \cite{Shapiro.book}.
A specification $S$ covered by $P$ is called in \cite{Apt93}
a supported Herbrand interpretation of $P$.

Now we are ready to present a sufficient condition for semi-completeness,
which together with Prop.\,\ref{prop:semi-compl}
provides a sufficient condition for completeness.

\begin{theorem}[Semi-completeness]
\label{th:completeness}%
If a specification $S$ is covered  by a program $P$ 
then $P$ is semi-complete w.r.t.~$S$.
\end{theorem}
The proof
is presented in Appendix
\ref{appendix}.  %
Note that the sufficient condition of Th.\,\ref{th:completeness}
 is equivalent to $S\subseteq T_P(S)$;
 the latter implies $S\subseteq {\it g f p}(T_P)$.
It is also equivalent to $S$ being a model of ONLY-IF$(P)$,
a theory used in defining the completion of $P$
(see e.g.\ \cite{Doets} or \cite{DBLP:journals/tplp/DrabentM05} for a definition).

\begin{example}
\label{ex:split-compl}%
We show that program SPLIT from Ex.\,\ref{ex:split-corr} is complete w.r.t.
the specification from Ex.\,\ref{ex:split-corr2},
\[
S_{\rm SPLIT} = 
\left\{\,
\begin{array}{l}
 s([\seq[2n]t], [t_1,\cdots,t_{2n-1}], [t_2,\cdots,t_{2n}] ),
\\
 s([\seq[2n+1]t], [t_1,\cdots,t_{2n+1}], [t_2,\cdots,t_{2n}] )
\end{array}
   \,
   \left|
   \,
   \begin{array}l
      n\geq0,\\ \seq[2n+1] t\in \HU  \\
   \end{array}
\right.\right\}
,
\]
where
$[t_k,\cdots,t_l]$ denotes the list $[t_k,t_{k+2},\ldots,t_l]$,
for $k,l$ being both odd or both even.

Atom
$s([\,],[\,],[\,])\in S_{\rm SPLIT}$ is covered by clause (\ref{split1}).
For $n>0$,
any atom
$A= \linebreak[3]
s([\seq[2n]t], \linebreak[3]
      [t_1,\cdots,t_{2n-1}],\linebreak[3] [t_2,\cdots,t_{2n}] )$
is covered by an instance of  (\ref{split2}) with a body
$B= \linebreak[3]
     s([t_2,\ldots,t_{2n}], \linebreak[3]
     [t_2,\cdots,t_{2n}] ,\linebreak[3] [t_3,\cdots,t_{2n-1}]
 )$.
Similarly, 
for $n\geq0$ and any atom
$A= s([\seq[2n+1]t],\linebreak[3]
         [t_1,\cdots,t_{2n+1}],\linebreak[3] [t_2,\cdots,t_{2n}] )$,
the corresponding body is
$B= s([t_2,\ldots,t_{2n+1}],\linebreak[3] [t_2,\cdots,t_{2n}],
 \linebreak[3]     [t_3,\cdots,t_{2n+1}] )$.
In both cases, $B\in S_{\rm SPLIT}$.
(To see this, rename each $t_i$ as $t'_{i-1}$.)
So $S_{\rm SPLIT}$ is covered by SPLIT.
Thus SPLIT is semi-complete w.r.t.\ $S_{\rm SPLIT}$, by
Th.\,\ref{th:completeness}.

Now by Prop.\,\ref{prop:semi-compl} the program is complete, as it is
recurrent under the level mapping
\[
  | \, s(t,t_1,t_2) | = |t|,
\quad \mbox{where} \quad
\begin{array}[t]{l}
  |\, [h|t]\, | = 1+|t|,    \\
  |f(\seq t)| = 0 \ \mbox{ where $n\geq0$ and $f$ is not }[\ | \ ],    \\
\end{array}
\]
(for any ground terms $h,t,\seq t$, and any function symbol $f$). 

\end{example}

\begin{example}
\label{ex:append-compl}
Consider program APPEND and specification $S_{\rm APPEND}^0$ from
Ex.\,\ref{ex:spec-append}. 
We encourage the reader to check that each
$A\in S_{\rm APPEND}^0$ is covered (by APPEND w.r.t.\ $S_{\rm APPEND}^0$).
Hence by Th.\,\ref{th:completeness} the program is semi-complete
(w.r.t.\ $S_{\rm APPEND}^0$).  Now its completeness follows by 
Prop.\,\ref{prop:semi-compl}, as the program is recurrent under the level
mapping $|app(t,t_1,t_2) | = |t|$, where $|t|$ is as in the
Ex.\,\ref{ex:split-compl}. 

\end{example}

The notion of semi-completeness is tailored for finite programs.  An SLD-tree
for an infinite program may be infinite, but with all branches finite.  Then
a slight generalization of Th.\,\ref{th:completeness}
applies (Appendix, Prop.\,\ref{prop:compl:infi}, Ex.\,\ref{ex:prop:compl:infi}).

\newcommand{\propcompletenessinfinite}
  {%
    If all the atoms from a specification $S$ are covered w.r.t.~$S$
     by a program $P$,
    and there exists an SLD-tree for $P$ and a query $Q$ with no infinite
    branches then $P$ is complete for $Q$.
  }
\subsection{Proving completeness directly}
\label{sec:completeness-directly}
Semi-completeness and completeness are declarative properties of programs.
However 
the sufficient conditions for completeness of
Prop.\,\ref{prop:semi-compl}(\ref{prop:semi-compl:term}) and
\ref{prop:semi-compl}(\ref{prop:semi-compl:accept}) are not declarative.  
The former refers explicitly to an operational notion of SLD-tree.
The latter employs the notion of acceptable program, which depends
on the order of atoms in program clauses, and is related to LD-resolution;
this should not be considered as a purely logical reading of formulae.
Only condition (\ref{prop:semi-compl:recu}) of 
Prop.\,\ref{prop:semi-compl} can be considered declarative, 
as it solely depends on logical reading of program clauses.
However, this condition is often inapplicable, as many practical programs are
not recurrent. 
Note that the sufficient condition for semi-completeness of 
Th.\,\ref{th:completeness} is declarative.

We are interested in declarative reasoning about logic programs.
So we present a declarative sufficient condition for completeness.
Moreover,
it is applicable in cases where completeness does not follow from
semi-completeness. 
We first generalize the notion of level mapping
 by allowing that it is a {\em partial}
function $|\ |\colon \HB\partto \NN$ assigning natural numbers to some atoms.

\begin{definition}
  A ground atom $H$ is {\bf recurrently covered}
  by a program $P$   w.r.t.\ a specification $S$ and a level mapping 
$|\ |\colon \HB\partto \NN$
  if $H$ is the head of a ground instance  
  $
  H\gets \seq B
  $
($n\geq0$)
of a clause of the program, such that
$|H|, |B_1|, \ldots |B_n|$ are defined, 
$\seq B\in S$,
and $|H|>|B_i|$ for all $i=1,\ldots,n$.

\end{definition}

For instance, given a specification 
$S = \{\, p(s^i(0))\mid i\geq0 \,\}$,
atom $p(s(0))$ is recurrently covered by a program
$\{\, p(s(X))\gets p(X).\ \ p(X)\gets p(X). \}$ under a level mapping for which 
$|p(s^i(0))|=i$.
No atom is recurrently covered
by a program \mbox{$\{\, p(X)\gets p(X). \}$}.
Obviously, if $H$ is recurrently covered by $P$ then it is covered by $P$.

The following theorem
 is a reformulation of Th.\,6.1 from \cite{Deransart.Maluszynski93}
(derived there within a sophisticated theory relating attribute
grammars and logic programs).

\begin{theorem}[Completeness 2]
\label{th:completeness:recu}%
If, under some level mapping $|\ |\colon \HB\partto \NN$,
 all the atoms from a specification $S$ are
recurrently covered by a program $P$ w.r.t.\ $S$ 
then $P$ is complete w.r.t.\ $S$.
\end{theorem}

\begin{proof}\nopagebreak
Note first that $|A|$ is defined for all $A\in S$.
We show that $A\in \M_P$ for each  $A\in S$, by induction on $|A|$.
Let $n\geq0$,
assume that $B\in \M_P$ whenever $B\in S$ and $|B|<n$.
Take an $A\in S$ with $|A|=n$.
As $A$ is recurrently covered, there is a clause instance $A\gets\seq[m] B$
($m\geq0$)
with $B_i\in S$ and $|B_i|<n$ for each $i$.  So  $\seq[m] B\in \M_P$, 
and $A\in \M_P$.
\end{proof}

\begin{example}
\label{ex:graph}
Consider a directed graph $E$.  As a specification for a program describing
reachability in $E$, take 
$S = S_p\cup S_e$,
where
$
S_p = \{\, p(t,u) \mid
 \mbox{there is a path from } t \mbox{ to } u \mbox{ in } E \,\}
$,
and
$
S_e = \{\, e(t,u) \mid (t,u)\mbox{ is an edge in } E \,\}
$.
Let $P$ consist of 
a procedure $p$:
\[
  \begin{array}[t]{l}
    p(X,X). \\
    p(X,Z) \gets e(X,Y),\, p(Y,Z). \\
  \end{array}
\]  
and a procedure $e$,
which is a set of unary clauses describing the edges of
the graph.  Assume the latter is complete w.r.t.\ $S_e$.
Notice that when $E$ has cycles then 
infinite SLD-trees cannot be avoided, and
completeness of $P$ cannot be shown by 
Prop.\,\ref{prop:semi-compl}.

To apply Th.\,\ref{th:completeness:recu}, 
let us define a level mapping for the elements of $S$ such that
$|e(t,u)| = 0$ and
$|p(t,u)|$ is the length of a shortest path in $E$ from $t$ to $u$
 (so $|p(t,t)|=0$).
Consider an atom $p(t,u)\in S$ where $t\neq u$.
Let 
$t=t_0,\seq t=u$ be a shortest path from $t$ to $u$.
Then $e(t,t_1),p(t_1,u)\in S$,
$|p(t,u)|=n$, $|e(t,t_1)|=0$, and $|p(t_1,u)|=n-1$.
Thus $p(t,u)$ is recurrently covered by $P$ w.r.t.\ $S$ and $|\ |$.
The same trivially holds for the remaining atoms of $S$, as they are
instances of unary clauses of $P$.
By Th.\,\ref{th:completeness:recu}, $P$ is complete w.r.t.\ $S$.

Note that $P$ is not correct w.r.t.\ $S$, unless \HU is the set of nodes of $E$
(as $p(t,t)\not\in S$ when $t$ is not a node of the graph).
This, again, illustrates usefulness of approximate specifications.
\end{example}

Completeness of the presented completeness proof methods is discussed in
Appendix
\ref{appendix:compl}. 
In particular it is shown that, in a certain sense, the sufficient
condition for semi-completeness of Th.\,\ref{th:completeness}
 is a necessary condition for completeness.

\section{Pruning}
\label{sec:pruning}

Pruning some parts of SLD-trees is often used to improve efficiency of programs.
In Prolog it is implemented by using the cut, the if-then-else construct, or
by built-ins, like {\tt once/1}.  
Pruning preserves the correctness of a logic program, it also preserves
termination under a given selection rule,
but may violate the program's completeness.  
Consequently,
the answers of a program $P$ executed with pruning may be correct w.r.t.\ 
a specification $S$ for which $P$ is not correct;
a related programming technique is known as ``red cuts''
\cite{Sterling-Shapiro-short}.

This section deals with reasoning about program answers in the presence of
pruning. 
The focus is on proving completeness.
We begin with a formalization of pruning.
Then we introduce a sufficient condition for completeness,
and present examples.  
The next subsection identifies
a class of cases where the sufficient condition does not
hold (despite pruned trees being complete),
and shows how a completeness proof can be obtained anyway.
Then we present a sufficient condition for completeness dealing also with
infinite pruned trees.
We conclude with an approach to prove correctness of the answers of pruned trees.
In some examples it is assumed that the reader is familiar with the
operational semantics of the cut in Prolog.

\subsection{Pruned SLD-trees}

By a {\bf pruned SLD-tree} for a program $P$ and a query $Q$ we mean a tree
with the root $Q$
which is a connected subgraph of an SLD-tree for $P$ and $Q$.
By an  {\em answer} of a pruned SLD-tree we mean the computed answer of a
successful SLD-derivation which is a branch of the tree.
We will say that a pruned SLD-tree $T$ with root $Q$
is {\bf complete} w.r.t.\ a specification $S$
if, for any ground  $Q\theta$,
$S\models Q\theta$ implies that $Q\theta$ is an instance of an answer of $T$.
Informally, such a tree produces all the answers for $Q$ required by $S$.

%
%
%
%


\newcommand{\myellipse}{%
    \begin{pgfpicture}{-1.5cm}{-.5cm}{1.5cm}{.5cm}
      \pgfsetxvec{\pgfpoint{.5cm}{0cm}}
      \pgfsetyvec{\pgfpoint{0cm}{.5cm}}
      \pgfsetlinewidth{1pt}
      \pgfmoveto{\pgfxy(-3,0)}
      \pgfcurveto{\pgfxy(-3,-1)}{\pgfxy(3,-1)}{\pgfxy(3,0)}
      \pgfstroke
      \pgfmoveto{\pgfxy(3,0)}
      \pgfsetdash{{3pt}{3pt}}{0pt}
      \pgfcurveto{\pgfxy(3,1)}{\pgfxy(-3,1)}{\pgfxy(-3,0)}
      \pgfstroke
    \end{pgfpicture}
}
  \newcommand{\dblue}{\color[rgb]{0,0,.5}}
  \newcommand{\dgreen}{\color[rgb]{0,0.5,0}}



\newcommand{\mytreepicture}[1][{}]
{%
\begin{pgfpicture}{-2.5cm}{-.5cm}{2.5cm}{1.7cm}
      \pgfsetxvec{\pgfpoint{.5cm}{0cm}}
      \pgfsetyvec{\pgfpoint{0cm}{.5cm}}
  \pgfputat{\pgfxy(0,3)}{\pgfbox[center,center]
    {{\footnotesize\ldots,}$\underline A${\footnotesize,\ldots}}
  }
  \pgfline{\pgfxy(-0.6,2.5)}{\pgfxy(-5,0)}
  \pgfline{\pgfxy(-0.2,2.5)}{\pgfxy(-.5,0)}
  \begin{pgfscope}
    \pgfsetdash{{3pt}{1.5pt}}{0pt}
    \pgfline{\pgfxy(0.2,2.5)}{\pgfxy(0.5,0)}
    \pgfline{\pgfxy(0.6,2.5)}{\pgfxy(5,0)}
  \end{pgfscope}


  \pgfputat{\pgfxy(-4.9,.7)}
           {\pgfbox[center,center]{\mbox{\small\footnotesize$\dgreen\Pi_{#1}$}}}
  \pgfputat{\pgfxy(-1.1,1)}{
    \pgfbox[center,center]{\dgreen
      \begin{pgfmagnify}{.7}{.6}
          \myellipse
      \end{pgfmagnify}
      }
  }
  \pgfputat{\pgfxy(-1.2,1.2)}{{\pgfbox[center,center]{$\cdots$}}}
  \pgfputat{\pgfxy(+2.2,0.6)}{{\pgfbox[center,center]{$\cdots$}}}

  \pgfputat{\pgfxy(-2.9,1.9)}
           {\pgfbox[center,center]{\mbox{\small\footnotesize$\dblue P$}}}
  \pgfputat{\pgfxy(-2.4,1.3)}{\dblue
      \begin{pgfmagnify}{.8}{.5}
          \myellipse
      \end{pgfmagnify}
  }      
  \pgfputat{\pgfxy(2.5,-.5)}
           {\pgfbox[center,center]{\mbox{\small\footnotesize pruned}}}  
  \pgfputat{\pgfxy(-2.5,-.5)}
           {\pgfbox[center,center]{\mbox{\small\footnotesize not pruned}}}  
\end{pgfpicture}
} 

  %

%
%
%

%
%
%
%
%
  To facilitate reasoning about the answers of pruned SLD-trees,
  we will view pruning as applying only certain clauses 
  while constructing the children of a node.
  So we introduce subsets $\seq\Pi$ of $P$.  The intention is that
  for each node the clauses of one $\Pi_i$ are used.
  Programs $\seq\Pi$ may be not disjoint.%
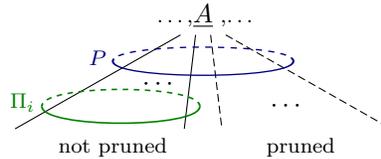
\[
{\mytreepicture[i]}
\]

\begin{definition}
Given programs $\seq{\Pi}$ ($n>0$),
a {\bf c-selection rule} is a function%
\
assigning to a query $Q'$ an atom $A$
in $Q'$ and one of the programs $\emptyset,\seq{\Pi}$.

A {\bf csSLD-tree} (cs for clause selection) for a query
$Q$ and programs
$\seq\Pi$, via a c-selection rule $R$,
is constructed as an SLD-tree, but
for each node its children are constructed using the program selected by the
c-selection rule.
An {\em answer} of the csSLD-tree is the answer of a successful derivation
which is a branch of the tree.
\end{definition}

A c-selection rule may choose the empty program, thus 
making a given node a leaf.
Note that program $\Pi_i$ selected for a node $Q'$ may contain clauses
whose heads do not unify with the selected atom of $Q'$
(thus inapplicable in constructing the children of the node).
A csSLD-tree for $Q$ and $\seq\Pi$ is a pruned SLD-tree
for $Q$ and $\bigcup_i\Pi_i$.   Conversely, for each pruned SLD-tree $T$
for $Q$ and a (finite) program $P$ there exist $n>0$, and
$\seq\Pi\subseteq P$ such that 
 $T$ is a csSLD-tree for $Q$ and $\seq\Pi$.

    It may be needed that different atoms or programs are selected for repeated
    occurrences of the same query.
    So, formally, a c-selection rule is a function of a sequence of queries
    (a path in a tree) $Q,\ldots,Q'$.
    Also, strictly speaking, not an atom, but 
    an occurrence of an atom in $Q'$ is selected.
In some cases it may be convenient to treat a c-selection rule as a function
of not only the sequence of queries, but also of 
(the not pruned part of) the subtree rooted at $Q'$
(as often the decision about pruning is taken after a part of the subtree has
been computed).

\begin{example}
\label{ex:overlap}
Consider a Prolog program
\[
\newcommand{\prologif}{\mathrel{\texttt{:-}}}
\tt
c(L1,L2) \prologif\ m(X,L1), m(X,L2), !.
\qquad\qquad
m(X, [X|\_] \, ). \qquad\qquad
m(X, [\_|L] \, )\prologif\ m(X,L). 
\]
which is a logic program $P$ with the cut added.
Let  $\Pi_1 = P$ and \ $\Pi_2=\{m(X, [X|\_] \, ).\}$.
Consider a query $Q=c( [\seq[k] t], [\seq u])$, the SLD-tree $T$ for $Q$ via the
Prolog selection rule, and the csSLD-tree $T'$ which is
$T$ pruned due to the cut.
If $Q$ does not succeed then $T'=T$.
Otherwise $T'$ is $T$ without all the nodes to the right of the leftmost
successful branch.
So in $T'$, for a node $Q'$ with an $m$-atom selected, $\Pi_2$ is selected
whenever there is a success in the subtree rooted in the first child of $Q'$;
otherwise $\Pi_1$ is selected.
Note that the single cut may prune
children of many nodes of $T$.

The c-selection rule of $T'$ can be expressed as a function of a node:
  $\Pi_2$ is selected for a node 
$Q_i= m(X, [\SEQ t i k]), m(X, [\seq u])$ in $T'$
 if $t_i$ is unifiable with some element of $[\seq u]$.
Otherwise $\Pi_1$ is selected for $Q_i$.
For a node of the form $Q'_j=  m(t, [\SEQ u j n])$,
program $\Pi_2$ is selected if $t$ is unifiable with $u_j$,
otherwise $\Pi_1$ is selected.
\end{example}

In this example, determining the c-selection rule corresponding to the cut
was not straightforward.  
Simpler cases are provided by Examples 
\ref{ex:adam}, \ref{ex:pruning1}, \ref{ex:pruning2} below.

\subsection{Completeness of pruned trees}

We begin with an example illustrating some of the difficulties related to
completeness of pruned csSLD-trees.  Then, after introducing some auxiliary
notions, a sufficient condition for completeness is presented.
A few example proofs follow.  
We conclude with an example of a completeness proof in a case for which the
sufficient condition is not applicable directly.

\begin{example}
\label{ex:prune}%
We show that completeness of each of $\seq\Pi$ is not sufficient for
completeness of a csSLD-tree for $\seq\Pi$.
   Consider a program $P$:
\vspace*{-1.5\abovedisplayskip}
\[
\hspace{-3em}
 \begin{minipage}{.35\textwidth}%
    \begin{eqnarray}
      &&
\label{Excl1}
    q(X)\gets p(Y,X). \\
      &&
\label{Excl2}
     p(Y,0).
    \end{eqnarray}
  \end{minipage}
\qquad\quad
  \begin{minipage}{.4\textwidth}
    \begin{eqnarray}  
\label{Excl3}  &&
     p(a,s(X))\gets p(a,X).      \\
\label{Excl4} &&
      p(b,s(X))\gets p(b,X).
    \end{eqnarray}
  \end{minipage}
\]
and programs $\Pi_1= \{(\ref{Excl1}), (\ref{Excl2}), (\ref{Excl4})\}$,
$\Pi_2= \{(\ref{Excl1}), (\ref{Excl2}), (\ref{Excl3})\}$,
    As a specification for completeness consider
    $S_0=\{\,q(s^j(0)) \mid j\geq0 \,\}$.
  Each of the programs  $\Pi_1, \Pi_2,P$ is complete w.r.t.\ $S_0$.
    Assume a c-selection rule $R$ choosing alternatively $\Pi_1,\Pi_2$ along
    each branch of a tree. 
    Then the csSLD-tree for $q(s^j(0))\in S_0$ via $R$ (where $j>2$) 
    has no answers,
    thus the tree is not complete w.r.t.\ $S_0$.
    Similarly the csSLD-tree for $q(X)$ is not complete.
\end{example}

Consider programs $P,\seq\Pi$ and specifications $S,\seq S$, such that 
$P\supseteq\bigcup_{i=1}^n\Pi_i$ and $S=\bigcup_{i=1}^n S_i$.
The intention is that each $S_i$ describes which answers are to be produced by
using $\Pi_i$ in the first resolution step.
We will call $\seq {\Pi}$, $\seq S$ a {\bf split} (of $P$ and $S$).

\begin{definition}
\label{def:suitable}
Let  $\S =\seq {\Pi}$, $\seq S$ be a split, and $S=\bigcup S_i$.
Specification $S_i$ is {\bf suitable} for an atom $A$
w.r.t.\ \S
when no instance of $A$ is in $S\setminus S_i$.
(In other words, when $ground(A)\cap S \subseteq S_i$.)
We also say that a program $\Pi_i$ is {\bf suitable} for $A$ w.r.t.\ \S
when $S_i$ is.

A c-selection rule is {\bf compatible} with \S if for each non-empty query
$Q'$ it selects an atom $A$
and a program $\Pi$, such that 

\quad
-- $\Pi\in\{\seq\Pi\}$ is suitable for $A$ w.r.t.\ \S, or

\quad
-- none of $\seq\Pi$ is suitable for $A$ w.r.t.\ \S and  $\Pi=\emptyset$
(so $Q'$ is a leaf).

A csSLD-tree for $\seq {\Pi}$ via a c-selection rule compatible with \S
is said to be {\bf weakly compatible} with \S.
The tree is {\bf compatible} with \S when 
for each its nonempty node some $\Pi_i$ is selected.
\end{definition}

  The intuition is that 
  when $\Pi_i$ is suitable for $A$
  then $S_i$ is a fragment of $S$ sufficient to deal with $A$.
  It describes all the answers for query $A$ required by~$S$.

The reason of incompleteness of the trees in Ex.\,\ref{ex:prune}
may be understood 
as selecting a $\Pi_i$ not suitable for the selected atom.
Take  $\S=\Pi_1,\Pi_2, S_0\cup S_1',S_0\cup S_2'$,
where
\mbox
{$S_1' = \{\,p(b,s^i(0)) \mid i\geq0 \,\}$} and
$S_2' = \{\,p(a,s^i(0)) \mid i\geq0 \,\}$.
In the incomplete trees,
$\Pi_1$ is selected for an atom $A=p(a,u)$,
or $\Pi_2$ is selected for an atom $B=p(b,u)$ (where $u\in\TU$).
However
$\Pi_1$ is not suitable for $A$ 
whenever $A$ has an instance in $S$
(as then $ground(A)\cap S \subseteq S_2'$, 
thus $ground(A)\cap S \not\subseteq S_0\cup S_1'$);
similarly for $\Pi_2$ and $B$.

  When  $\Pi_i$ is suitable for  $A$ then if
    each atom of $S_i$ is covered by $\Pi_i$
  (w.r.t.\ $S$)
  then using for $A$ only the clauses of $\Pi_i$ does not impair
  completeness w.r.t.~$S$:

\newcommand{\thcompletenesspruned}[1]
{%
    Let $P\supseteq\bigcup_{i=1}^n\Pi_i$ (where $n>0$) be a program, 
    $S=\bigcup_{i=1}^n S_i$ a specification, and 
    $T\!$ a csSLD-tree for $\seq\Pi$.
    If
    \begin{enumerate}
    \item 
    \label{prop:cssld.complete.cond0}
        for each $i=1,\ldots,n$, 
        all the atoms from $S_i$ are covered by $\Pi_i$ w.r.t.\ $S$, and 
    \item
    \label{prop:cssld.complete.cond1}
        $T$ is compatible with $\seq\Pi,\seq S$,
    \item 
    \label{prop:cssld.complete.cond2}
    \begin{enumerate}
      \item 
      \label{prop:cssld.complete.cond2a}
      #1
      \item 
      \label{prop:cssld.complete.cond2b}
          $P$ is recurrent, or
      \item 
      \label{prop:cssld.complete.cond2c}
            $P$ is acceptable and 
          $T$ is built under the Prolog selection rule 
    \end{enumerate}
    \end{enumerate}
    then $T$ is complete w.r.t.\ $S$.
} %

\begin{theorem}
\label{th:completeness:pruned}%
\label{prop:cssld.complete}%
\thcompletenesspruned{         $T$ is finite, or}
\end{theorem}
See Appendix \ref{appendix} for a proof.
Note that in \ref{prop:cssld.complete.cond2c}
the program may be acceptable w.r.t.\ a specification distinct from $S$.

\begin{example}
\label{ex:adam}
A Prolog program \cite{clocksin-mellish-1ed}
\[
{\tt
n o p(a d a m,0) \texttt{ :- !.}
\qquad
n o p(eve,0) \texttt{ :- !.}
\qquad
n o p(X,2).
}
\]
is an example of difficulties and dangers of using the cut in Prolog.
(For its (in)correctness, see Ex.\,\ref{ex:pruning:correctness}.)
Due to the two occurrences of the cut, for an atomic query $A$
only the first clause with the head unifiable with $A$ will
be used.
The program can be seen as logic program $P=\Pi_1\cup\Pi_2\cup\Pi_3$ executed
with pruning, where (for $i=1,2,3$)
$\Pi_i$ is the $i$-th clause of the program
with the cut removed.
The intended meaning is $S=S_1\cup S_2\cup S_3$, where 
$S_1 = \{ n o p(a d a m,0) \}$, 
$S_2 = \{ n o p(eve,0) \}$, and
$S_3 = \left\{ \rule{0pt}{1.9ex}
     n o p(t,2)\in\HB \mid t\not\in\{a d a m,eve\} \right\}$.
Note that all the atoms from $S_i$ are covered by $\Pi_i$ (for $i=1,2,3$).

Let \S be $\Pi_1,\Pi_2,\Pi_3,S_1, S_2, S_3$.  
Consider a query $A=n o p(t,Y)$ with a ground $t$.
If $t=a d a m$ then only $\Pi_1$ is suitable for $A$ w.r.t.\ \S,
as $ground(A)\cap S = S_1$.
If $t=eve$ then only  $\Pi_2$ is suitable.
For $t\not\in\{a d a m,eve\}$ the suitable program is $\Pi_3$.
So for a query $A$ the pruning due to the cuts in the program results in 
selecting a suitable $\Pi_i$, and the obtained
csSLD-tree is compatible with \S. 
By Th.\,\ref{th:completeness:pruned}
the tree is complete w.r.t.\ $S$.

For a query $n o p(X,Y)$ or $n o p(X,0)$ only the first clause, i.e.\ $\Pi_1$,
is used.  However $\Pi_1$ is not suitable for the query (w.r.t.\ \S),
and the csSLD-tree is not compatible with \S.  The premises of 
Th.\,\ref{th:completeness:pruned} do not hold, and 
the tree is not complete (w.r.t.\ S).
\end{example}

\begin{example}
\label{ex:pruning1}%
The following program SAT0 is a simplification of a fragment of the SAT solver 
of \cite{howe.king.tcs-shorter}
discussed in [\citeNP{Drabent12.iclp};\,\citeyearNP{drabent.corr.2012}].
Pruning is crucial for the efficiency and usability of the original program.
\begin{quote}
\mbox{}\hspace{-3em}
  \begin{minipage}[t]{.5\textwidth}%
\vspace{-3\abovedisplayskip}
    \begin{eqnarray}
      &&
\label{SCcl1}
      p( P\mydash P, \, [\,] ).
          \\
      &&
\label{SCcl2}
      p( V\mydash P, \, [B|T] ) \gets q( V\mydash P, \, [B|T] ).   \\
      &&
\label{SCcl3}
      p( V\mydash P, \, [B|T] ) \gets q( B, \, [V\mydash P|T] ).   
    \end{eqnarray}
  \end{minipage}
  \begin{minipage}[t]{.4\textwidth}%
\vspace{-3\abovedisplayskip}
    \begin{eqnarray}
      &&
\label{SCcl4}
      q( V\mydash P, \, \mbox{\LARGE\_}\, )\gets V=P.                 \\
      &&
\label{SCcl5}
      q(\, \mbox{\LARGE\_}\, , \, [A|T] )\gets p( A, T).              \\
      &&
\label{SCcl6}
      P = P.
    \end{eqnarray}
  \end{minipage}  
\end{quote}
The program is complete w.r.t.\ a specification
\[
S =
\left.\left
\{\, 
\begin{array}{l}
  p(t_0\mydash u_0,[t_1\mydash u_1,\ldots, t_n\mydash u_n]),      \\
  q(t_0\mydash u_0,[t_1\mydash u_1,\ldots, t_n\mydash u_n])    
\end{array}
\, \right| \,
\begin{array}{l}
    n\geq0,
    \ t_0,\ldots,t_n,u_0,\ldots,u_n \in\mathbb{T},
    \\
    t_i=u_i \mbox{ for some } i\in \{0,\ldots,n\}\,
\end{array}
\right
\}\cup S_=
\]
where $\mathbb T = \{{\it false}, {\it true}\}\subseteq\HU$, and
$S_= = \{\, t{=}t \mid t\in\HU \,\}$.
We omit a completeness proof,
mentioning only that SAT0 is recurrent w.r.t.\ a level mapping
$|p(t,u)| = 2|u|+2$, 
$|q(t,u)| = 2|u|+1$, $|{=}(t,u)|=0$,
where $|u|$ for $u\in \HU$ is defined in Ex.\,\ref{ex:split-compl}.

The first case of pruning is related to redundancy within 
(\ref{SCcl2}), (\ref{SCcl3});
 both
$\Pi_1={\rm SAT0}\setminus \{(\ref{SCcl3})\}$ and 
$\Pi_2={\rm SAT0}\setminus \{(\ref{SCcl2})\}$
are complete w.r.t.\ $S$.
For any selected atom at most one of (\ref{SCcl2}), (\ref{SCcl3}) is
used, and the choice is dynamic.
As the following reasoning is independent from this choice, we omit further
details.

So in 
such pruned SLD-trees
the children of each node are constructed using one of programs
$\Pi_1, \Pi_2$.
Thus they are csSLD-trees for $\Pi_1, \Pi_2$.
They are compatible with $\S= \Pi_1, \Pi_2, S, S$
(as $\Pi_1, \Pi_2$ are trivially suitable for any $A$, due to 
 $S\setminus S_i=\emptyset$ in Df.\,\ref{def:suitable}).
Each atom of $S$ is covered w.r.t.\ $S$ both by $\Pi_1$ and $\Pi_2$.
As SAT0 is recurrent,
by Th.\,\ref{th:completeness:pruned}
each such tree is complete w.r.t.\ $S$.

\end{example}

\begin{example}
\label{ex:pruning2}%
We continue with program SAT0 and specification $S$ from the previous
example, and add
 a second case of pruning.
When the selected atom is of the form $A=q(s_1,s_2)$ with a ground $s_1$ then
only one of 
clauses (\ref{SCcl4}), (\ref{SCcl5}) is needed --
(\ref{SCcl4}) when $s_1$ is of the form $t\mydash t$, and  (\ref{SCcl5})
otherwise.
The other clause can be abandoned without losing the completeness w.r.t.\ $S$.%
\footnote
{
 The same holds for $A$ of the form  $q(s_{11}\mydash s_{11},s_2)$, or
 $q(s_{11}\mydash s_{12},s_2)$ with non-unifiable $s_{11}$, $s_{12}$.
 The pruning is implemented using the if-then-else construct in
 Prolog: \
{\tt q(V-P,[A|T]) :-  V=P ->  true ; p(A,T).}
 (And\,
{\tt p(V-P,[B|T]) :- non{}var(V) -> q(V-P,[B|T])  ; q(B,[V-P|T])}\,
  implements the first case of pruning.)
}

 Actually, SAT0 is included in a bigger program, say $P={\rm SAT0}\cup\Pi_0$.
 We skip the details of $\Pi_0$, let us only state that $P$ is recurrent,
 $\Pi_0$ does not
 contain any clause for $p$ or for $q$, and that
 $P$ is complete w.r.t.\ a specification $S'=S\cup S_0$ where
 $S_0$ does not contain any $p$- or $q$-atom.
 (Hence each atom of $S_0$ is covered by $\Pi_0$ w.r.t.\ $S'$.)

To formally describe the trees for $P$ resulting from both cases of pruning,
consider the following programs and specifications:
\[
\begin{tabular}[t]{l @{\ \ }l}
   $\Pi_1 =  \{(\ref{SCcl1}),(\ref{SCcl2})\}$, \
    $\Pi_2 =  \{(\ref{SCcl1}),(\ref{SCcl3})\}$,
&       $S_1=S_2=S\cap \{\,p(s,u)\mid s,u\in\HU\,\}$,
\\
  $\Pi_3 =  \{(\ref{SCcl4})\}$,
&       $S_3=S\cap\{\,q(t\mydash t,s)\mid t,s\in\HU \,\}$,
\\
   $\Pi_4 =  \{(\ref{SCcl5})\}$,
&       $S_4= S\cap\{\,q(t\mydash u,s)\mid  t,u,s\in\HU,t\neq u \,\}$,
\\
   $\Pi_5 =  \{(\ref{SCcl6})\}$,
&  $S_5 = S_=$.
\end{tabular}
\]
Let \S be $\Pi_0,\ldots,\Pi_5,S_0,\ldots,S_5$.
Each atom from $S_i$ is covered by $\Pi_i$ w.r.t.\ $S'$ (for $i=0,\ldots,5$).
For each $q$-atom $A$ with its first argument ground,
$\Pi_3$ or $\Pi_4$ is suitable.
(Both $\Pi_3,\Pi_4$ are suitable when $ground(A)\cap S = \emptyset$.)
For each remaining atom from \TB
(at least) one of programs $\Pi_0,\Pi_1,\Pi_2,\Pi_5$ is suitable.

Consider a pruned SLD-tree $T$ for $P$ (employing the two cases of pruning
described above).  
Assume that each $q$-atom selected in $T$ has its first argument ground.
Then $T$ is a csSLD-tree compatible with \S.
From Th.\,\ref{th:completeness:pruned} it follows that $T$ is complete
w.r.t.\ $S'$.

The restriction on the selected $q$-atoms is
implemented by means of Prolog delays.  This is done in such a way 
that, for the intended initial queries, floundering is avoided
\cite{howe.king.tcs-shorter}
(i.e.\ an atom is selected in each query).
So the obtained pruned trees are as $T$ above, and
 the pruning preserves completeness of the program.
\end{example}

\subsubsection{Sufficient condition not applicable directly}

Th.\,\ref{th:completeness:pruned} has an additional consequence.
Not only $T$ but also any its subtree is complete w.r.t.\ $S$
(as the subtree satisfies the premises of Th.\,\ref{th:completeness:pruned}).
So the theorem is not (directly) applicable when a csSLD-tree is complete,
but pruning removes some answers for a query within the tree
(cf.\ Ex.\,\ref{ex:compl.pruning.inapplicable} below).
In Ex.\,\ref{ex:pruning.proved} we show how, at least in some such cases,
completeness can be proved by employing a family of specifications.

\begin{example}
\label{ex:compl.pruning.inapplicable}
Consider the logic
  program $P$ from Ex.\,\ref{ex:overlap}.
Predicate $c$ of $P$ is complete w.r.t.\ a specification
\[
S_c = \left\{\, c(l_1,l_2)\in\HB \: \left|
            \begin{tabular}{@{\ }l@{}}
              $l_1,l_2$  are lists with\\ a common element
            \end{tabular}\!
        \right.\,\right\}.
\]
The pruning does not impair completeness for ground queries of the form
$Q=c(l_1,l_2)$.
To apply Th.\,\ref{th:completeness:pruned} to prove the completeness,
each atom of $S_c$ has to be covered by the first clause of $P$ w.r.t.\ 
a specification $S$.
Hence
$S  \supseteq S_c\cup S_m$ is needed where
$S_m = \{\, m(t_i,[\seq t])\in\nolinebreak\HB  \mid n>0,\ 1\leq i \leq n \,\}$.

Consider now a pruned tree for such $Q$, and its subtree $T''$ rooted at
$Q'= m(X,l_1), m(X,l_2)$.  
The subtree is not complete whenever $l_1,l_2$ have more than one common
element. 
Thus, for each such case, Th.\,\ref{th:completeness:pruned} is inapplicable
for any superset of $S_c$.%
\footnote{
  It is not applicable in a more general case of $l_1,l_2$ having a common
  element, which is not the last element of $l_1$.
  The pruned tree for such $Q$ contains a node with 
  $A= m(Y, [\SEQ t i k])$ and $\Pi_2=\{\,m(X, [X|\_] \, ).\}$ selected
  (where $\SEQ t i k$ are ground, and $t_i$ is a member of $l_2$).
  For such atom $A$, a suitable specification $S_2$ is a superset of
  $ground(A)\cap S \supseteq \{\,m(t_j, [\SEQ t i k]) \mid i\leq j \leq k\,\}$.
    But such $S_2$ cannot be covered by $\Pi_2$,
  unless all $ \SEQ t i k$ are equal.
  So the premises of Th.\,\ref{th:completeness:pruned} are not satisfied.

  Note that this reasoning holds for any representation of the considered
  pruned trees as csSLD-trees, possibly employing some other $n,\seq\Pi$.
  } %

\end{example}

\begin{example}
\label{ex:pruning.proved}
The previous example shows a case where completeness of pruned trees 
w.r.t.\ a specification cannot be shown by Th.\,\ref{th:completeness:pruned}.
However the completeness can be proved by
 employing a family of specifications.
Consider the same program $P$ and specification $S_c$.
For each $c(l_1,l_2)\in S_c$ we construct a specification $S'$
for which $P$ is semi-complete,  $c(l_1,l_2)\in S'$,
and Th.\,\ref{th:completeness:pruned} is applicable.
Informally, $S'$ may be understood as the subset of $\M_P$ needed to obtain
$c(l_1,l_2)$, 
by means of the clauses involved in a pruned tree for $c(l_1,l_2)$.

Let $l_1=\nolinebreak{}[\seq[k] t],\ l_2=[\seq u]$.
So there exist a term $t$ and indices $g,h$ such that $t=t_g=u_h$,
for $i=1,\ldots,g-1$ no $t_i$ is a member of $l_2$, and 
$u_j\neq t$  for $j=1,\ldots,h-1$.  Let
$S' = \{\, c(l_1,l_2) \,\} \cup S'_m$ where
\[
S'_m =
\{\,
  m(t,[\SEQ t i k]) \mid 1\leq i \leq g
\,\} \cup
\{\,
  m(t,[\SEQ u j n]) \mid 1\leq j \leq h
\,\}.
\]  
Let $\S = \Pi_1,\Pi_2,S_1,S_2$, where $\Pi_i$ are as in Ex.\,\ref{ex:overlap},
$S_1=S'$, 
$S_2= \{\,
  m(t,[\SEQ t g k]),\ \linebreak[3]
  m(t,[\SEQ u h n]) 
\,\}
$.
All the atoms from $S_i$ are covered by $\Pi_i$ w.r.t.\ $S'$ (for $i=1,2$).
Consider the csSLD-tree $T'$ for the (ground) query $Q = c(l_1,l_2)$.
From the description in Ex.\,\ref{ex:overlap} of the c-selection rule,
it follows that whenever $\Pi_2$ is selected together with an atom $A$ in $T'$
then $A=  m(X,[\SEQ t g k])$ (for some variable $X$) or $A= m(t,[\SEQ u h n])$.
In both cases $ground(A)\cap S' \subseteq S_2$.
Thus $T'$ is compatible with \S and, as it is finite, 
by Th.\,\ref{th:completeness:pruned} $T'$ is complete w.r.t.\ $S'$.

We showed that, for each $Q\in S_c$, the pruned tree is complete w.r.t.
$\{Q\}$.
Hence for each ground query of the form $Q=m(s_1,s_2)$ 
the csSLD-tree for $Q$ is complete w.r.t.\ $S_c$.

\end{example}

\subsubsection{Another sufficient condition for completeness}

By generalizing
Th.\,\ref{th:completeness:recu} to pruned SLD-trees, we obtain a sufficient
condition for completeness, which does not require that the trees are finite.
See Appendix \ref{appendix} for a proof.

\newcommand{\propcompletenesspruned}
{%
    Let $P \supseteq \bigcup_{i=1}^n\Pi_i$ (\/$n>0$) be a program, 
    $S=\bigcup_{i=1}^n S_i$ a specification,
     $|\ |\colon \HB\partto \NN$ a level mapping,  and 
    $T\!$ a csSLD-tree for $\seq\Pi$.
    If
    \begin{enumerate}
    \item
        $T$ is compatible with $\seq\Pi,\seq S$,
    \item 
        each atom from $S_i$ is recurrently covered by $\Pi_i$ w.r.t.\ $S$
        and $|\ |$,
        for each $i=1,\ldots,n$, 
    \end{enumerate}
    then $T$ is complete w.r.t.\ $S$.
} %

\begin{proposition}
\label{prop:cssld.complete2}%
\propcompletenesspruned
\end{proposition}

\begin{example}
Consider the program $P$, specification $S$ and the level mapping $|\ |$ from
Ex.\,\ref{ex:graph}.
 Consider a c-selection rule $R$ which selects the
first atom of the query, and if the atom is of the form
$p(t,t)$ it selects $\Pi_1=\{\,  p(X,X)  \,\}$,
if it is of the form $p(t,u)$ with distinct ground $t,u$ then it selects
$\Pi_2=\{\,     p(X,Z) \gets e(X,Y),\, p(Y,Z)   \,\} $,
otherwise the whole program is selected,  $\Pi_3=P$.
Let $S_1=\{\, p(t,t) \mid t \mbox{ is a node of } E \,\}$,
\linebreak[3]%
\mbox{$S_2=\{\, p(t,u)\in S \mid t\neq u \,\}$}, $S_3=S$.
Note that $R$ is compatible with $\S= \Pi_1,\Pi_2,\Pi_3, S_1,S_2,S_3$
 (as in the first case 
 $ground(p(t,t))\cap S \subseteq S_1$, and in the second case
 $\{p(t,u)\}\cap S \subseteq S_2$).
{\sloppy\par}

From  Ex.\,\ref{ex:graph} it follows that each atom from $S_i$
is recurrently covered by $\Pi_i$ w.r.t.\ $S$, for $i=1,2,3$.
Thus by Prop.\,\ref{prop:cssld.complete2},
any csSLD-tree for $P$ via $R$ is complete w.r.t.\ $S$.
Note that the tree may have infinite branches.

Assume that the initial queries are ground.  Then, 
under the Prolog selection rule,
all the selected $p$-atoms are ground.  So the c-selection rule $R$ 
can be implemented in Prolog by adding the cut to the first clause of~$P$:
{\tt p(X,X):-!.}
However when the csSLD-tree is infinite, not the whole tree may searched by
Prolog due to the depth-first search strategy.

\end{example}

\subsection{Correctness of answers of pruned trees}
The answers of pruned trees are, obviously, logical consequences of the program.
However, it may happen that, due to pruning, not all such consequences are
obtained.
This feature is sometimes used in Prolog programming; it is called ``red cuts''
\cite{Sterling-Shapiro-short}.
In this section we show how to characterize more precisely the answers in
such cases.  We present a sufficient condition for the answers of csSLD-trees to 
be correct (cf.\ Df.\,\ref{def:corr:compl:aux}),
 w.r.t.\ a specification for which the program may be not correct.
Informally, the idea is to apply the condition of
Th.\,\ref{th:correctness} 
only to those instances of program clauses which were employed in the pruned
trees.

\begin{theorem}
\label{th:cssld.correct}
Let $T$ be a csSLD-tree for a query $Q$ and programs $\seq\Pi$, and $S$ be a
specification. 
If for each $i=1,\ldots n$
\[
\begin{minipage}{.92\textwidth}
\raggedright
  $\seq[m]B\in S$ implies   $H\in S$, 
  for each $(H\gets\seq[m]B)\in ground(\Pi_i)$ ($m\geq0$)
  where $H$ is an instance of some atom $A$ selected together with  $\Pi_i$ 
  in $T$,
\end{minipage}
\]
then each answer $ Q\theta$ of $T$ is correct w.r.t.\ $S$
(i.e.\ $S\models Q\theta$).
\end{theorem}

\begin{proof}
Assume the premises of the theorem.
We have to show that  $S\models Q\theta$.
Answer $Q\theta$ is obtained out of a successful branch of the tree.
The branch represents a derivation with queries $\SEQ Q 0 k$,
mgu's $\seq[k]\theta$, and renamed clauses  $\seq[k]C$,
where $Q_k$ is empty, $Q_0=Q$, and  $Q\theta = Q\SEQC \theta 1 k$.
By induction on $k-i$ we obtain $S\models Q_i\SEQC\theta {i+1} k$ for
$i=k, k{-}1,\ldots, 0$.
The base case is obvious.
For the inductive step, assume $S\models Q_i\SEQC\theta {i+1} k$.
Let $Q'\in ground(Q_{i-1}\SEQC\theta i k)$.
Let an atom $A$ and a program $\Pi_j$ be selected for $Q_{i-1}$.
So $C_i\in\Pi_j$.
A ground instance $Q''$ of $Q_i\SEQC\theta {i+1} k$ can be obtained out of 
$Q'$ by replacing an instance $H$ of $A$ by a sequence of atoms $\seq[m]B$
($m\geq0$), where  $(H\gets\seq[m]B)\in ground(C_i)$. 
As each atom of $Q''$ is in $S$, so is each atom of $Q'$.
\end{proof}

\begin{example}
\label{ex:pruning:correctness}
Let us continue Ex.\,\ref{ex:adam}.  The program is obviously not correct
w.r.t.\  $S$.  However, we show that for queries of the form 
$A = n o p(t,Y)$, where $t$ is ground,
 the answers of pruned trees are correct w.r.t.\ $S$.

The condition of Th.\,\ref{th:cssld.correct} obviously holds for $\Pi_1,\Pi_2$.
Consider  $\Pi_3$, it is selected together with $A = n o p(t,Y)$ only if
$t\not\in\{a d a m,eve\}$.  Then any ground instance $H$ of both $A$ and 
the head $n o p(X,2)$ of the clause from $\Pi_3$ is in $S$.
Thus the premise of Th.\,\ref{th:cssld.correct} is satisfied.

On the other hand, it is not satisfied for e.g.\ a query $A'= n o p(Y,2)$ 
(due to its instance $H = n o p(a d a m, 2) \not\in S$, $H\in ground(\Pi_3)$
and $\Pi_3$ is selected together with $A'$).
Note that
the pruned tree for such $A'$ is not correct w.r.t.\ $S$.
\end{example}

\section{Relations with declarative diagnosis}
\label{sec:dd}
We now discuss the relationship between program diagnosis, and proving
correctness and completeness of programs.
We first introduce declarative diagnosis of logic programs, 
and compare it with the proof methods of Sections
\ref{sec:correctness}, \ref{sec:completeness};
details irrelevant for this paper are omitted.
Then we show how diagnosis can be performed by attempting a correctness or
completeness proof of the program concerned.
Finally we discuss a major drawback of declarative diagnosis,
and show how it can be avoided by employing approximate specifications.

\subsection{Declarative diagnosis}
\label{sec:dd.sub}
Declarative diagnosis,
called sometimes algorithmic, or declarative, debugging,
was introduced by 
\citeN{Shapiro.book}
(see also \cite{DNM89,DBLP:conf/acsc/Naish00} and references therein).
It provides methods to locate
the reasons of incorrectness or incompleteness of programs. 
A diagnosis algorithm starts with a {\em symptom}, obtained from testing the
program; the symptom shows incorrectness or incompleteness of the program
w.r.t.\ its specification $S$. 
The symptom is
an answer $Q$ such that \mbox{$S\notmodels Q$}, or a
query $Q$ for which the computation terminates but some answers required by $S$ are not produced.
(A less general notion of incompleteness symptom is an atom $A\in S$ for
which the program finitely fails.)
The located error turns out to be a fragment of the program
 which violates our sufficient condition for
correctness or, respectively, semi-completeness.

\paragraph{Incorrectness diagnosis}

In declarative diagnosis,
the reason for incorrectness is an incorrect instance of a program clause.
An incorrect clause instance
is one which violates the sufficient condition of Th.\,\ref{th:correctness}. 
Obviously, by Th.\,\ref{th:correctness}
if the program is incorrect then such clause must exist.
  On the other hand,
  a violation of the sufficient condition does not imply that the
  program is incorrect (cf.\ Ex.\,\ref{ex:weak-spec} in the Appendix).
Incorrectness diagnosis cannot be made more precise; we cannot say which part
of the incorrect clause is wrong, as different modifications of various parts
of the incorrect clause may lead to a correct one.

Technically, an incorrectness diagnosing algorithm works by constructing a
proof tree for an atomic symptom $Q$, and then searching the tree for an incorrect
clause instance.
Notice that the actions performed by the algorithm boil down to checking 
the sufficient condition for correctness of Th.\,\ref{th:correctness},
but only for some clause instances -- those involved in producing the symptom.

\begin{example}[Incorrectness diagnosis]
\label{ex:inc-diag}%
Consider a buggy insertion sort program \cite[Sect.\,3.2.3]{Shapiro.book}:
\[
\newcommand{\dred}{\color[rgb]{.7, 0, 0}}
\mbox
{\tt\small
  \begin{tabular}{@{}l}
	isort([],[]).
\\
	isort([X|Xs],Ys) :- isort(Xs,Zs), insert(X,Zs,Ys).
\\[1ex]
       	insert(X,[],[X]).
\\
	insert(X,[Y|Ys],[Y|Zs]) :- Y > X, \mylonger{insert(X,Ys,Zs).}
\\
	insert(X,[Y|Ys],[X,Y|Ys]) :- X =< Y.
  \end{tabular}%
} %
\]
(where {\tt >}, {\tt =<} are arithmetic built-ins of Prolog \cite{Apt-Prolog}).
Procedure ${\it isort}$ 
should define the sorting relation for lists of integers.
Procedure ${\it insert}$ 
should describe inserting an element into a sorted list, producing
a sorted list as a result.  The program computes an incorrect answer 
$A = {\it isort}( [2,1,3], [2,3,1] )$.  
(Formally, $S\notmodels A$, where $S$ is a
specification known by the user; we do not formalize $S$ here.)
Thus $A$ is a symptom of incorrectness.
The diagnoser constructs a proof tree
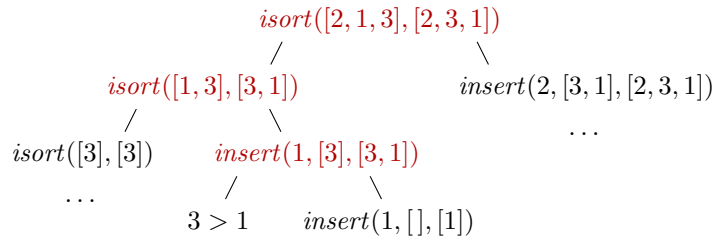
\[
    \newcommand{\dred}{\color[rgb]{.7, 0, 0}}
\newcommand{\arca}{
                   \begin{pgfpicture}
                         \pgfpathmoveto{\pgfpointorigin}
                         \pgfpathlineto{\pgfpoint{1ex}{2ex}}
                         \pgfusepath{stroke}
                   \end{pgfpicture}
                  }
\newcommand{\arcb}{
                   \begin{pgfpicture}
                         \pgfpathmoveto{\pgfpointorigin}
                         \pgfpathlineto{\pgfpoint{-1ex}{2ex}}
                         \pgfusepath{stroke}
                   \end{pgfpicture}
                  }
    \newlength{\myextra}
    \setlength{\myextra}{1ex}
    \begin{array}[t]{@{}cc}
      \begin{array}{c}
        \dred
            {\it isort}( [2,1,3], [2,3,1] )        \\[.5ex]
         \arca  \hspace{7em}\ \arcb
      \end{array}
     \qquad\ 
    \\[\myextra]
      \begin{array}[t]{c@{}}
         {\dred
          {\it isort}([1,3],[3,1])
         }
         \\[.5ex]
         \arca  \hspace{5em} \arcb
      \end{array}
      \qquad\qquad\qquad
      \begin{array}[t]{@{}c}
        \mbox{${\it insert}(2,[3,1],[2,3,1])$}
        \\  
        \raisebox{-1ex}\ldots
      \end{array}
      \\[\myextra]
      \begin{array}[t]{c@{}}
        {\it isort}([3],[3]) \\
        \raisebox{-1ex}\ldots
      \end{array}
      \qquad
      \begin{array}[t]{c@{}}
        {\dred
          {\it insert}(1,[3],[3,1])
        }
        \\[.5ex]
          \arca  \hspace{5em} \arcb \quad
      \end{array}
    \qquad\qquad\qquad \quad  \qquad\qquad\qquad\qquad
    \\[\myextra]
    \qquad
    \qquad\qquad
      3>1
    \qquad
      {\it insert}(1,[\,],[1])
    \qquad\qquad\qquad\qquad\qquad\qquad
    \end{array}
\hspace{-3em}
\]
Starting from the root of the proof tree,
the diagnosis algorithm asks the user whether the
children of the current node are correct w.r.t.\ the user's specification.
When all the answers are YES, the
current node together with its children are an incorrect clause instance --
the error is found.  An answer NO means that a new incorrect node is found,
and it becomes the current node.

For the tree above, the questions are first asked about the two children of the
root.  The child $ A'={\it isort}([1,3],[3,1])$ is found incorrect,
formally $S\notmodels A'$.
(The question about the second child
is discussed in the next example.) 
Queries about the children of $A'$
lead to identifying ${\it insert}(1,[3],[3,1])$  as incorrect.  
As both its children are found correct, the algorithm returns 
${\it insert}(1,[3],[3,1]) \gets 3\mathop{>}1,   {\it insert}(1,[\,],[1])$
as an incorrect clause instance.
(The nodes found incorrect are marked in the diagram.)

Note that more precise locating of the error is impossible, The wrong
clause may be corrected for instance by swapping $X$ and $Y$ in $Y>X$, or
in the first and last atom of the clause.
Thus we cannot determine which fragment of the clause is erroneous.

\end{example}

\paragraph{Incompleteness diagnosis}
There exist a few versions of incompleteness diagnosis algorithms
(\cite{DBLP:journals/ngc/Naish92} provides a comparison).
Some of algorithms
 start with a symptom which is a finitely failed specified atom
$A\in S$.  
Some use, instead, a query $Q$ for which the SLD-tree is
finite, and some answers required by the specification $S$ are missing.
As the reason for incompleteness,
a not covered specified atom is found, say $p(\ldots)\in S$.
Alternatively, the result of the diagnosis is an atom $B = p(\ldots)$
which has a not covered specified instance $B\theta\in S$.
In both cases procedure $p$ is identified as erroneous.

The
existence of a not covered specified atom violates the sufficient condition for
semi-completeness of Th.\,\ref{th:completeness}.
Conversely, if the program is not complete for a query $Q$ with a finite SLD-tree
then it is not semi-complete and, by Th.\,\ref{th:completeness},
 there must exist a not covered specified atom.
 Again, violating the sufficient condition for completeness does not imply
 that the program is not complete.
Roughly speaking, the specification may be not sufficiently general,
cf.\ Ex.\,\ref{ex:weak-spec}.
Similarly to incorrectness, incompleteness diagnosis cannot be made more
precise:
A whole procedure is located,
and it cannot be determined which fragment of the
procedure is wrong, or whether a clause is missing.

We skip an example, and details of incompleteness diagnosis algorithms, see
e.g.\ \cite{Shapiro.book,DNM89,DBLP:journals/ngc/Naish92}.
We only point out a relation between completeness proving and some of such
algorithms (like that of \cite{Shapiro.book}). 
The actions performed by such algorithm boil down
to checking 
the sufficient condition for semi-completeness of Th.\,\ref{th:completeness}.
In contrast to completeness proofs, the condition is checked only for some
specified atoms, they are instances of atoms selected in the SLD-tree
producing the symptom.

\paragraph{Diagnosis by proof attempts}
In a sense, the proof methods presented in this paper supersede the
declarative diagnosis methods.
An attempt to prove a buggy program to be correct (complete) results in
violating the corresponding sufficient condition for some clause (specified
atom).
Any error located by diagnosis will also be found by a proof attempt;
moreover
no symptom is needed, and all the errors are found.
For instance, in this way
the author found an error in an early version of one of the examples 
from \cite{Drabent12.iclp}.
However the sufficient condition has to be checked for all the
(instances of the)
clauses of the program (when proving correctness), or for all specified atoms
(in the case of completeness).
In contrast, diagnosis takes into account only the clauses
(respectively, atoms) related to the computation that produced the symptom.

\subsection{A main drawback of declarative diagnosis}
\label{sec:dd-drawback}
Declarative diagnosis requires the user to know an exact specification (a
single intended model) of the program.  The program has to be (made) correct
and complete w.r.t.\ this specification.  This is a serious difficulty in
diagnosing actual programs
  \cite[Sect.\,26.8]{DNM89},
  \cite{DBLP:conf/acsc/Naish00}.
Some diagnoser queries,
 like ``is\linebreak[3] $append( [a], b, [a|b] )$ correct'',
may be confusing and hardly possible to answer,
 as the programmer often does not know some details of the
intended model,
like those related to applying $append$ to non lists.
In the author's opinion this was a main reason of the lack of
acceptance of declarative diagnosis in practice.

\pagebreak[3]
\begin{example}[Exact specification unfeasible]
\label{ex:inc-diag2}
\nopagebreak
The previous example provides a diagnoser query which is hardly possible to answer.
Determining correctness of $B ={\it insert}(2,[3,1],[2,3,1])$ is
confusing, as the user does not know how insertion into unsorted lists should
behave.
From a formal point of view, the user does not know an exact specification of 
${\it insert}$.

Actually, both cases of atom $B$ being correct or incorrect are reasonable.
They result in locating various errors,
and lead to various versions of the final, corrected program.
If the specification for ${\it insert}$ is
  \footnote{
  An (exact) specification for the remaining predicates is obvious:
\\[.5ex]
$
  \begin{array}{c}
  \left\{\,
      {\it sort}( l_1, l_2) \in\HB \ \left| \ 
      \mbox
        {\begin{tabular}{@{}l@{}}
            $l_1$ is a list of integers, \\
            $l_2$ is a sorted permutation of $l_1$ \\
          \end{tabular}%
        }
      \right.\right\}
\cup
 \left\{\,
      \begin{array}{@{}r@{}}
          {>}( i_1, i_2) \in\HB, \\
          {=}{<}( j_1, j_2) \in\HB
      \end{array}
        \ \left| \   
      \mbox
        {\begin{tabular}{@{}l@{}}
            $i_1,i_2,j_1,j_2$ are integers,\\
             $i_1>i_2$, $j_1 \leq j_2$
          \end{tabular}%
        }
      \right.\right\}.
  \end{array}
$
}   %
\[
 S^0_{\it insert} = \left\{\,
    {\it insert}( n, l_1, l_2) \in\HB\ \left| \ 
    \mbox
      {\begin{tabular}{@{}l@{}}
          $l_1, l_2$ are sorted lists of integers, \\
          ${\it elms}(l_2) = \{n\}\cup {\it elms}(l_1)$
        \end{tabular}
      }
    \right.\right\},
\]
where ${\it elms}(l)$ is the multiset of the elements of list $l$,
then $S^0_{\it insert}\notmodels B$ and 
an erroneous clause instance is eventually found within the subtree rooted in
$B$.  Correcting the program
would lead to a procedure ${\it insert}$ which is correct w.r.t.\ 
$S^0_{\it insert}$ but inefficient, as it assures that two of
its arguments are sorted lists of integers.

If the specification is 
\[
S_{\it insert} = \left\{\,
    {\it insert}( n, l_1, l_2) \in\HB \ \left| \ 
    \mbox
      {\begin{tabular}{@{}l@{}}
          if $n\in\ZZ$ and $l_1$ is a sorted lists of integers, \\
          then ${\it insert}( n, l_1, l_2)\in S^0_{\it insert}$
        \end{tabular}
      }
    \right.\right\}
\]
then $S_{\it insert}\models B$, and the diagnosis continues as in
Ex.\,\ref{ex:inc-diag}.
Each of the alternative corrections described there leads to a program
with ${\it insert}$  correct w.r.t.\ $S_{\it insert}$.
One of the programs is
the ``right'' one -- Program 3.21 of \cite{Sterling-Shapiro-short}
(with 
${\it insert}(X,[Y|Y s],[Y|Zs]) \gets \mbox{$X > Y$}, {\it insert}(X,Y s,Zs)$
as the corrected clause).

Notice that neither $S_{\it insert}$ nor $S^0_{\it insert}$ is an exact
specification of procedure ${\it isort}$ in the ``right'' program.
Providing its exact specification is not so obvious.
The  ``right'' procedure is correct w.r.t.\ $S_{\it insert}$,
and complete w.r.t.\ $S^0_{\it insert}$.
Any such procedure, combined with procedure ${\it isort}$ from
Ex.\,\ref{ex:inc-diag}, would 
implement a correct and complete predicate  ${\it isort}$.
The programmer should have the freedom of choosing (and changing) the
actual semantics of ${\it insert}$, guided by e.g.\ the efficiency of the
program. 
An exact specification would be counterproductive, at least at earlier stages
of program development.

So for this procedure one should take $S_{\it insert}$ as the specification for
correctness and $S^0_{\it insert}$ as the specification for completeness,
and apply them respectively in incorrectness and incompleteness diagnosis.
\end{example}

The discussion in the example suggests a natural way to overcome the
difficulty due to declarative diagnosis being based on a single intended
model of a program:

\begin{remark}
Declarative diagnosis should employ approximate specifications.
  In the incorrectness diagnosis the specification for correctness
should be used, and in the incompleteness diagnosis -- the specification for
completeness. 
\end{remark}

The problem with exact specifications
was first dealt with in \cite{Pereira86-short}.
In that approach a specification also describes which atoms are
{\em admissible goals}, i.e.\ are allowed to be selected in SLD-trees.
Thus specifications and diagnosis are not declarative, as they refer to the
operational semantics.
\citeN{DBLP:conf/acsc/Naish00} introduces a three valued approach, in which
a specification   classifies
  each ground atom as {\em correct},  {\em erroneous} or  {\em inadmissible}.
  For such specifications, three-valued declarative debugging algorithms
  are presented.
  From our point of view, the set of non-erroneous atoms can be understood as
  a specification for correctness, and the set of correct atoms as a
  specification for completeness.  However introducing debugging algorithms
  based on a three-valued logic seems to be an unnecessary complication.
In the approach proposed here the standard two-valued diagnosis algorithms remain unchanged.

\section{Related work}
\label{sec:related}

\paragraph{Approximate specifications}
The difficulties caused by a restriction to exact specifications
have been recognized by a few authors.
\citeN[and the references therein]{Naish.Sondergaard13.short}
 propose employing a 3-valued or 4-valued logic.
The logical value {\bf i} (inadmissible)
is given to the atoms that, in our approach, are in the difference
$S_{corr}\setminus S_{\it c o m p l}$
between the specification for correctness and that for completeness.

\citeN{Apt-Prolog} and 
\citeN{Deville} present approaches in which it is possible to deal with a suitable
subset of $\M_P$ instead of $\M_P$ itself. 
\citeN{Apt-Prolog} employs preconditions and postconditions,%
\footnote{
    Technically, pre- and postconditions are sets of atoms closed under
    substitution. 
}
and a sufficient condition that,
 for any atomic query $A$,
if $A\in pre$ then 
$A\theta\in post$ for any computed answer $A\theta$
(where $pre$ and $post$ is, respectively, a pre- and postcondition).
In this way $\M_P$ is in a sense replaced by $post$,
more precisely by the set $\M_P\cap pre$, called $\M_{(pre,post)}$.
(See also {\it\nameref{par:related:proofmethods}} below.)
Eg.\ for program APPEND the precondition may be
$pre =  \{\, app(k,l,m) \in\TB \mid
                \mbox{$l$ or $m$ is a list}
\,\}
$,
and the postcondition may describe exactly the list concatenation relation:
$post = S_{\rm APPEND}^0$, in the notation of  Ex.\,\ref{ex:spec-append}.
This gives $\M_{(pre,post)} = S_{\rm APPEND}^0$.
See \cite{DBLP:journals/tplp/DrabentM05} for further comparison.

In \cite{Deville} the notions of correctness and completeness of programs 
involve, for each predicate of the program,
the intended relation (on ground terms) and the so-called domain.
Let us represent them as Herbrand interpretations, the intended relations as 
$rel\subseteq\HB$ and the domains as $d o m\subseteq\HB$.
Then correctness (called there partial correctness) of a definite clause
program $P$ is equivalent to 
$\M_P\cap d o m\subseteq rel \cap d o m$, and completeness means inclusion in
the other direction.
For the APPEND example, $ rel = S_{\rm APPEND}^0$ and
$d o m =   \{\, app(k,l,m)\in\HB \mid
                k,l,m\mbox{ are lists}
\,\}
$.

Sometimes 
it is claimed that programs like APPEND
should be modified, so that the resulting 
program defines ``the right'' relation  \cite{Naish96,Apt93,Apt95}.
For the case of APPEND, the new program should define exactly 
the list concatenation.
This can be done by modifying its non-recursive clause
so that its variable is restricted to be bound to a list.
The clause becomes $app([\,],L,L) \gets {\it list}(L)$,
and an appropriate definition of ${\it list}$ is added.
The program is believed to be simpler to reason about, but
it is impractical due to its inefficiency.
\citeN[p.\,187]{Naish96} proposes that the meaning of the original program
should be defined as the standard meaning of the modified program.


The reader is encouraged to compare the treatment of APPEND in 
Examples \ref{ex:spec-append}, \ref{ex:append-corr}, \ref{ex:append-compl}
to that in the approaches outlined here.
To sum up, 
for reasoning about correctness and completeness of definite programs,
the approach based on approximate specifications seems preferable,
as it is simpler than the approaches outlined above, 
refers only to the the basic notions of logic programming,
and seems at least as powerful.

\enlargethispage*{2ex}
\paragraph{Proof methods} 
\label{par:related:proofmethods}
Below we review previous work on proving correctness and
completeness of logic programs.
We omit most of approaches based on operational semantics.

Our definitions of correctness and completeness (Df.\,\ref{def:corr:compl})
are basically the same as those in \cite{Sterling-Shapiro-short}.
However no particular method of reasoning on correctness and completeness is
given there, and the example informal proofs are made by analyzing possible proof trees.
The correctness proving method of \cite{Clark79} used here
(Th.\,\ref{th:correctness})  should be
well-known, but is often neglected.
For instance, an important monograph \cite{Apt-Prolog} uses a method,%
\footnote{
It is basically the method of \cite{DBLP:conf/tapsoft/BossiC89}, (stemming
from  \cite{DM88},
a paper focused on non-declarative properties of programs).
The employed notion of correctness says that, 
given a precondition $pre$ and a postcondition $post$,
if the initial atomic query $A$ is in $pre$ then
in each LD-derivation each selected atom is in $pre$ and each
corresponding computed answer is in $post$.
Note that it depends on the order of atoms in a clause.
It implies a declarative property that, in our terminology,  the program is
correct w.r.t.\  $(\HB\setminus pre)\cup(post\cap\HB)$.
}
which is more complicated and not declarative.
It proves a certain property of LD-derivations, from which the
declarative property of program correctness follows.
See \cite{DBLP:journals/tplp/DrabentM05} for comparison and argumentation that
Clark's method is sufficient,
and for its generalization to programs with negation.
Surprisingly little has been done on reasoning about completeness.
For instance, it is absent from \cite{Apt-Prolog}%
\footnote{
Instead, 
   for a program $P$ and an atomic query $A$, 
   a characterization of the set of computed instances of $A$ is studied,
   in a special case of the set being finite and the answers ground
   \cite[Sect.\,8.4]{Apt-Prolog}.
   This is based on  finding  the least Herbrand model (of $P$ 
   or of a certain subset of $ground(P)$).
}
and \cite{hogger.book}.%
\footnote{
    The notion of completeness is defined in \cite{hogger.book},
    but no sufficient condition is discussed.
}
\,
\citeN{Kowalski85shorter-book}
discusses completeness, but the example proofs concern
only correctness.  
As a sufficient condition for completeness of a program $P$ 
he suggests
$P\vdash T_S$, where $T_S$ is a specification in a form of a logical theory.
(A symmetrical criterion $T_S\vdash P$ is used for correctness, notice its
similarity to that of Th.\,\ref{th:correctness}.) 
The condition seems
impractical as it fails when $T_S$ contains auxiliary
predicates, not 
occurring in $P$. 
It also requires that all the models of $P$ (including \HB) are models of the
specification.  
But it seems that such specifications often have a
substantially restricted class of models, maybe a single Herbrand
model, cf.\ \cite{Deville}.

\citeN{Stark98-short}
presented an elegant method of reasoning about a broad class of properties of
programs with negation,
executed under LDNF-resolutions.  
A tool to verify proofs mechanically was provided.
The approach involves a rather complicated induction scheme, so 
it seems impossible to apply the method informally by programmers.

\citeN{Deville} introduces a systematic approach to constructing programs.
Correctness and completeness of a program follows from construction.
No direct sufficient criteria, applicable to arbitrary programs, are given.  
The underlying semantics of Prolog programs is that of LDNF-resolution.
Note that the last two approaches are based on operational semantics,
and depend on the order of literals in clauses.  So they are not declarative.

\citeN{Deransart.Maluszynski93} present criteria for definite program completeness, 
in a sophisticated framework of relating logic programming and
attribute grammars.
We present, as Th.\,\ref{th:completeness:recu},
 their sufficient (and complete) criterion for completeness
in a simplified setting and with a short direct proof.

The method introduced here 
(Th.\,\ref{th:completeness} with Prop.\,\ref{prop:semi-compl})
 is a simplification of that from \cite{DBLP:journals/tplp/DrabentM05}.
Due to a restriction to specifications over Herbrand domains,
we could deal here with $ground(P)$ and covered ground atoms,
instead of a theory ONLY-IF$(P)$ and its truth in an interpretation,
 as in the former work.
Another difference is employing here a new notion of semi-completeness.
Also, 
the former work uses a slightly different notion of completeness
(cf.\ a remark in Sect.\,\ref{sec:corr+compl}),
and provides a generalization for programs with negation,
with a declarative semantics of \cite{Kunen87}.

Some ideas of the current work appeared in
[\citeNP{Drabent12.iclp};\,\citeyearNP{drabent.corr.2012}]
(see also {\em\nameref{par:applications}} in Sect.\,\ref{sec:discussion}).
This includes a weaker version of Th.\,\ref{prop:cssld.complete}
on completeness of pruned SLD-trees.
The author is not aware of any other former work on proving completeness
in presence of pruning,
or on proving correctness of the answers of pruned SLD-trees
(w.r.t.\ specifications for which the program is not correct).
Part of this work was presented in a conference paper 
\cite{drabent.lopstr14}.
\enlargethispage*{1.5ex}

For the related work on declarative diagnosis see Sect.\,\ref{sec:dd}.
Dealing with the occur-check problem
(see \cite{Apt-Prolog,Deransart.Maluszynski93} and the references therein)
is outside of the scope of this paper.
\pagebreak[3]
\section{Remarks}
\label{sec:discussion}%
\nopagebreak
\paragraph{Declarativeness}
\label{par:declarative}

In this paper we are mainly interested in declarative reasoning about logic
programs.
Correctness and completeness of logic programs are declarative properties.
However non declarative methods for dealing with correctness are often suggested,
cf.\ the
\hyperref[{par:related:proofmethods}]{previous section}.
If it were necessary to reason 
non declaratively about correctness or completeness of programs
then logic programming would not deserve to be called a declarative
programming paradigm. 
This paper presents declarative sufficient conditions for correctness 
and semi-completeness
 (Th.\,\ref{th:correctness}, \ref{th:completeness}).
Some of the conditions for completeness
are also declarative
(Prop.\,\ref{prop:semi-compl}(\ref{prop:semi-compl:recu}),
                Th.\,\ref{th:completeness:recu}).

The remaining sufficient conditions for completeness of
Prop.\,\ref{prop:semi-compl} are not declarative.   
Condition (\ref{prop:semi-compl:term}) involves termination,
condition (\ref{prop:semi-compl:accept}) depends on the order of atoms in
clause bodies (and implies termination).
Notice that declarative completeness proofs either imply termination
(when the program is shown to be recurrent,
Prop.\,\ref{prop:semi-compl}(\ref{prop:semi-compl:recu})), 
or require reasoning similar to that in termination proofs
(when showing that covered atoms are recurrently covered,
Th.\,\ref{th:completeness:recu}).
Notice also that in most practical cases termination has to be established anyway.
So proving completeness in two steps -- a declarative proof of
semi-completeness and a proof of termination --
may be a reasonable compromise
between declarative and non-declarative reasoning.
Moreover, in many cases termination proofs can be obtained automatically
(see
\cite{CodishT99-termination-shorter,MesnardB05-termination-tool-shorter,%
      Schneider-KampGST09-termination-shorter,NguyenSGS11-termination-shorter}
and references therein).
Semi-completeness alone may be a useful property, as it guarantees that whenever
the computation terminates, all the answers required by the
specification have been computed.

Note that the sufficient conditions of Section \ref{sec:pruning}
for completeness of pruned trees and correctness of their answers
also are declarative, except for conditions
\ref{prop:cssld.complete.cond2a},
\ref{prop:cssld.complete.cond2c} of Th.\ref{prop:cssld.complete}
(which are similar to conditions (\ref{prop:semi-compl:term}) and
 (\ref{prop:semi-compl:accept}) of Prop.\,\ref{prop:semi-compl},
 discussed above).
Proving condition \ref{prop:cssld.complete.cond2a}, i.e.\ termination,
can be partly automated, as mentioned above.

\paragraph{Granularity of proofs}
The sufficient conditions presented in this paper impose a certain granularity 
of proofs.  Correctness proofs deal with separate program clauses. 
Proofs of semi-completeness deal with whole program procedures
(to check that an atom $p(\ldots)$ is covered one has to consider all the
clauses for $p$).  For completeness, 
a certain level mapping has to be found, or program termination
is to be considered.  In both cases the whole program has to be taken into
account.

\paragraph{Interpretations as specifications}

Specifications which are interpretations
are somewhat limited \cite{DBLP:journals/tplp/DrabentM05}.  
Some program properties cannot be expressed.
The problem also concerns e.g.\ the approaches of
\cite{Apt-Prolog}, 
\cite{hogger.book}, \cite{Kowalski85shorter-book}, \cite{Deville},
and the declarative diagnosis.
The limitation is due to considering completeness w.r.t.\ a
specification which describes a single relation for a predicate symbol.
Such specification
 cannot express that e.g.\ for a given $t$
 there exists a $u$ such that $p(t,u)$ is the program's answer.
It has to explicitly state some (one or more) particular $u$.

In our framework an approximate specification describes a set
$\{\, I \mid S_{\it c o m p l} \subseteq I \subseteq S_{corr} \,\}$
of specifications;
its elements are the least Herbrand models of programs compatible with the
specification; the set has the least and the greatest element. 
We have just shown a case where an approximate specification should describe
a set with many minimal elements.
 \ 
(A symmetric case -- related to correctness -- is e.g.\ describing that
at most one of $p(t,u)$, $p(t,s)$ is allowed to be an answer of the program.)
A possible solution may be to introduce specifications in a form 
of logical theories.
Such a theory may include axioms like $\forall t\exists u.\, p(t,u)$,
and describe a set inexpressible in our framework.

\paragraph{Applications}
\label{par:applications}
We agree with the opinion of \cite{Apt93} that  ``unless the verification
method is easy and amenable to informal use, it will be ignored''.
We want to stress the simplicity and naturalness of the sufficient conditions
for correctness (Th.\,\ref{th:correctness}) and semi-completeness 
(Th.\,\ref{th:completeness},
the condition
 is a part of each discussed sufficient condition for completeness).
Informally, the former states that the clauses of a program should produce
only correct conclusions, given correct premises.  
  The latter states that each ground atom that should be produced by $P$ 
  can be produced by a clause of $P$
  out of atoms which should be produced by $P$.
The author believes that this is a way a competent programmer reasons about 
(the declarative semantics of) a logic program. 

The practical applicability of the methods presented here 
is illustrated by a larger example in
 [\citeNP{Drabent12.iclp}; \citeyearNP{drabent.corr.2012}].
The example demonstrates
 a systematic construction of a non-trivial Prolog program
(the SAT solver of \cite{howe.king.tcs-shorter}).  
Starting from a formal specification, a definite clause logic program
is constructed hand in hand with proofs of its correctness, completeness,
and termination under any selection rule.
The final Prolog program is obtained by adding control to the logic program
(delays and pruning SLD-trees).
Adding control preserves correctness and termination.
However completeness may be
violated by pruning, and by floundering related to delays.
By Th.\,\ref{prop:cssld.complete}, 
the program with pruning remains complete.%
\footnote{
In the former paper a weaker version of Th.\,\ref{prop:cssld.complete} has
been used, and one case of pruning is discussed informally. 
A proof covering both cases of pruning is illustrated here in
Ex.\,\ref{ex:pruning2}. 
}
Reasoning about floundering is outside of the scope of the current paper;
non-floundering of this program can be confirmed by a program analysis
algorithm \cite{king.non-suspension2008}.

The example shows how much of the programming task can be
  done declaratively, 
  without considering the operational semantics; how 
   ``logic'' could be separated from ``control.''  
  It is important that all the
  considerations and decisions about the program execution and efficiency
  are independent from those related
  to the declarative semantics: to the correctness and completeness
  of the final program. 
For the role of approximate specifications in this example, see 
{\em\nameref{par:approximate:comments}} in
Sect.\,\ref{par:approximate}.

\paragraph{Future work}
A natural continuation of the work presented here is 
generalization of the proof methods for programs with negation.
There are three semantics to deal with:
3-valued completion semantics of Kunen (Prolog with additional checks for
sound negation \cite{Kunen87,Doets}),
the well-founded semantics (the same with tabulation, as in XSB
\cite{wfsem,SwiftWarren.XSB.12}),
and the answer set semantics
 \cite[and the references therein]{BrewkaET.ASP.CACM11}.
The approach for Kunen semantics could possibly
stem from that of \cite{DBLP:journals/tplp/DrabentM05}.
Also generalization to constraint logic programming and to CHR (constraint
handling rules) could be of interest.

For reasoning about pruning, it should be useful to provide methods which
directly refer to the pruning constructs (like the cut) used in a program,
instead of employing the c-selection rule.
An initial attempt was presented in \cite{drabent.lopstr14}.
An interesting issue is introducing another form of specifications,
to overcome the limitation discussed earlier in this section.

An important task is formalization and automatizing of correctness and
completeness proofs; a first step is formalization of specifications.
On the other hand, further examples of proofs for practical programs are due, 
as are experiments with teaching programmers to informally use the proof methods
in practice.
\section{Conclusions}
\label{sec:conclusions}

This paper presents proof methods for proving correctness and
completeness of definite clause programs. 
The method for correctness  \cite{Clark79} is simple and natural.
It should be well-known, but is often neglected.
Little work has been done on proving completeness, and this is the main
subject of this paper
(Sect.\,\ref{sec:completeness}).
A simplification of the approach of \cite{DBLP:journals/tplp/DrabentM05} is
presented. 
  We introduce a notion of semi-completeness, for which the corresponding
  sufficient condition deals with program procedures separately,
  while for completeness the whole program has to be taken into account.
Semi-completeness and termination imply completeness.
In practice this means that if a semi-complete program has terminated then it
has produced all the answers required by the specification.
The presented sufficient condition for semi-completeness corresponds to a
natural, intuitive way of reasoning about programs.
It is a necessary condition for program completeness
(in a sense made precise in Appendix \ref{appendix:compl}).
We propose a few sufficient conditions for program completeness;
all of them involve the above mentioned condition for semi-completeness.
Also, sufficient conditions are given for completeness being preserved
under SLD-tree pruning (Sect.\,\ref{sec:pruning}).
This is augmented by a sufficient condition for correctness of the answers of
pruned SLD-trees
(applicable in describing the effects of so called ``red cuts''
\cite{Sterling-Shapiro-short}).

Logic programming can't be considered a declarative programming paradigm
unless there exist
declarative ways of reasoning about program correctness
and completeness
(i.e.\ reasoning 
which abstracts from the operational semantics).
Regrettably,
non-declarative methods are usually suggested (cf.\ Sect.\,\ref{sec:related}).
The presented methods for proving correctness and semi-completeness are
purely declarative (Sect.\,\ref{sec:discussion}),
 however some of the sufficient conditions for
completeness are not, as they refer, perhaps indirectly, to program termination.
Such methods may nevertheless be useful in practical declarative reasoning
about programs, as usually termination has to be established anyway.

  A larger example of applying the proof methods described here is provided in
  [\citeNP{Drabent12.iclp};\,\citeyearNP{drabent.corr.2012}].
  That paper presents a construction of a non-trivial
  program hand in hand with its correctness and completeness proofs.
  The construction shows how 
   ``logic'' can be separated from ``control'';
  how the reasoning about correctness and completeness can be separated from
  that related to the operational semantics, efficiency, etc
  (cf.\ {\em\nameref{par:applications}} in Sect.\,\ref{sec:discussion}).

  We point out advantages of approximate specifications
  (Sect.\,\ref{par:approximate}).  
  They are crucial for avoiding unnecessary
  complications in constructing specifications and proofs.
  It is often cumbersome and superfluous to exactly describe the
  relations computed by a program.
  Approximate specifications are natural in the process of program development:
  when starting construction of a program, 
  the relations it should compute are often known only approximately.
  This suggests an extension of the well-known paradigm of program
  development by transformations that preserve program semantics
  (see \cite{PettorossiPS10shorter} for references);
  one should also consider 
  transformations which preserve correctness and completeness
  w.r.t.\ an approximate specification
  (cf.\ {\em\nameref{par:approximate:comments}} in
    Sect.\,\ref{par:approximate}).

  In  Sect.\,\ref{sec:dd} we compared the proof methods with declarative
  diagnosis (algorithmic debugging).  
  Similarity was demonstrated between two ways of locating errors in
  programs: 
  by means of declarative diagnosis, and by a failure of an attempt to
  construct a correctness or completeness proof.
  We have also shown how approximate specifications help to remove
  a serious drawback of declarative diagnosis.

  We argue that the proof methods presented here are simple, and
  reflect a natural way of declarative thinking about programs
  (Sect.\,\ref{sec:discussion}).
We believe that they can be actually used -- at least at an informal level -- in
  practical programming; this is supported
  by examples.

\pdfbookmark[1]{Appendix}{bookmark:appendix}
  \appendix
    \section*{APPENDIX}
  \section{Completeness of proof methods}
  \label{appendix:compl}

  An important feature of a proof method is its completeness.
  Here we show in which sense the methods of proving correctness
  (Th.\,\ref{th:correctness}), semi-completeness (Th.\,\ref{th:completeness}),
  and completeness (Th.\,\ref{th:completeness:recu}) are complete.

  Obviously, we cannot require that if a program $P$ is correct (respectively,
  complete) w.r.t.\ a specification $S$ then a proof method applied to $P$ and
  $S$ shows the correctness (completeness).  The specification may provide
  insufficient information.  

\begin{example}
[Specifications too weak (strong) for a correctness (completeness) proof\/]
\label{ex:weak-spec}%
Program 
  $P=\{\, p(X)\gets q(X).\ q(a).\,\}$ is correct w.r.t.\ 
  $S=\{\, p(a), q(a), q(b) \,\}$.  However the sufficient condition from 
  Th.\,\ref{th:correctness} is violated: clause instance $p(b)\gets q(b)$
  has its body atom in $S$ and the head not in $S$.  We need a stronger
  specification $S'=\{\, p(a), q(a) \,\}\subseteq S$, for which $P$ is correct
  and the sufficient condition holds.

In Ex.\,\ref{ex:prune},
program $\Pi_1$ is complete w.r.t.\   $S_0=\{\,q(s^j(0)) \mid j\geq0 \,\}$,
but the sufficient condition for semi-completeness of
Th.\,\ref{th:completeness} does not hold:
some atoms of $S_0$ (actually all) are not covered by $P$ w.r.t.\ $S_0$.
To prove the completeness, we need a more
general specification
  $S=S_0\cup  \{\,p(b,s^j(0)) \mid j\geq0 \,\}$.
Now each atom $S$ is covered by $\Pi_1$ w.r.t.\ $S$, and the sufficient
condition holds.
(Similar reasoning applies to the completeness of $P$ and  $\Pi_2$ 
w.r.t.\   $S_0$.)

\end{example}

Theorems  \ref{th:compl.method.corr}, \ref{th:compl.method.compl} below
are from \cite{Deransart.Maluszynski93}.

  \begin{theorem}
  [Completeness of the method for correctness]
  \label{th:compl.method.corr}
  If a program $P$ is correct w.r.t.\ a specification $S$ then there exists 
  a stronger specification $S'\subseteq S$ such that the sufficient condition
  from Th.\,\ref{th:correctness} holds for $P$ and $S'$.
  \end{theorem}
  \begin{proof}
  Take $S'=\M_P$.
  \end{proof}

\begin{theorem}
[Completeness of the method for semi-completeness]%
\hspace{-1.5pt}%
If a program $P$ is semi-complete w.r.t.\ a specification $S$
then there exists a more general specification $S'\supseteq S$
such that the sufficient condition
 from Th.\,\ref{th:completeness} holds for $P$ and $S'$.
\sloppy
\end{theorem}

\begin{proof}
For $P$ and each atom $A\in S$ there exists a successful or infinite
SLD-derivation.   We consider ground instances of such derivations.
Let $D$ be a set of derivations for $ground(P)$, containing a successful or
infinite derivation for each $A\in S$.  
Let $S'$ be the set of atoms selected
in the derivations of $D$.  Each atom of $S'$ is covered by $P$.
\end{proof}

  Proving completeness of programs based on Prop.\,\ref{prop:semi-compl}
  and Th.\,\ref{th:completeness} is not complete.  It fails e.g.\ when a 
  program contains a void clause $A\gets A$.  However the method based on 
  Th.\,\ref{th:completeness:recu} is complete.

  \begin{theorem}
[Completeness of the method for completeness, Th.\,\ref{th:completeness:recu}\,]
  \label{th:compl.method.compl}
    If a program $P$ is complete w.r.t.\ a specification $S$ then 
    there exists a more general specification $S'\supseteq S$,
    and a level mapping $|\ |$, such that 
    the sufficient condition from Th.\,\ref{th:completeness:recu} holds for $P$,
    $S'$ and  $|\ |$.
  \end{theorem}
  \begin{proof}
  For each $A\in S\subseteq \M_P$ there exists a proof tree for $ground(P)$ with
  the root $A$.  Let $\cal T$ be the set of such trees and their subtrees.
  Let $S'$ be the set of their nodes.
  For a $B\in S'$,
  let $|B|$ be the minimal height of a tree from $\cal T$ with the root $B$.
  Consider such a minimal height tree.  Let $\seq B$ be the children of the
  root $B$ in the tree.  So $|B|>|B_i|$ for each $i$.
  As $\seq B\in S'$ and  $B\gets\seq B \in ground(P)$,
  atom $B$ is recurrently covered by $P$ w.r.t.\ $S'$ and $|\ |$.
  \end{proof}

  From the last theorem
   it follows that if $P$ is complete w.r.t.\ $S$ 
  then some $S'\supseteq S$ is covered by $P$.
  Thus the sufficient condition of Th.\,\ref{th:completeness}
  (applied to some $S'\supseteq S$)
  is a necessary condition for program completeness (w.r.t.\ $S$).

\section{Proofs}
\label{appendix}
This appendix contains the remaining proofs (for sections
\ref{sec:corr+compl}, \ref{sec:completeness} and \ref{sec:pruning}).
Let us 
note that the results of sections \ref{sec:corr+compl}--\ref{sec:pruning}
hold also for infinite programs, unless stated otherwise.%

\DeclareRobustCommand{\lemmaonMP}{\ref{lemma:MP}}
\subsection{The least Herbrand model and the logical consequences of a program
}
This section contains a proof of Lemma \lemmaonMP.

\medskip
{\sc Lemma \ref{lemma:MP}}. \
Let $P$ be a program, and $Q$ a query such that
\begin{enumerate}
\item 
 $Q$ contains exactly $k\geq0$ (distinct) variables, and
 the underlying language has (at least) $k$ constants not occurring in $P,Q$,
 or a non-constant function symbol not occurring in~$P,Q$.
\end{enumerate}
Then $\M_P\models Q$ iff $P\models Q$.
\medskip

\begin{proof}
The ``if'' case is obvious.
Assume $\M_P\models Q$.  
Let $\seq[k] V$ be the variables occurring in the query $Q$. 
Consider a substitution $\rho = \{ V_1/t_1,\ldots,V_k/t_k \}$ where $\seq[k]t$
are distinct ground terms whose main function symbols $\seq[k]f$ do not occur
in  $P\cup\{Q\}$.
Note that none of $\seq[k]f$ occurs in any computed answer for $P$ and $Q$.
As $M_P\models Q\rho$ and $Q\rho$ is ground, $P\models Q\rho$ by Th.\,4.30 of
\linebreak[3]\cite{Apt-Prolog}.

By completeness of SLD-resolution,
$Q\rho$ is an instance of some computed answer $Q\varphi$ for $P$ and $Q$:
$Q\rho=Q\varphi\sigma$.
So (for $i=1,\ldots,k$) $t_i=V_i\varphi\sigma$, $t_i=f_i(\ldots)$ and 
$f_i$ does not occur in $V_i\varphi$.
Thus each $V_i\varphi$ is a variable.
As $\seq[k] t$ are distinct,  $V_1\varphi,\ldots,V_k\varphi$ are distinct.
Thus $Q\varphi$ is a variant of $Q$ and, by soundness of
SLD-resolution, $P\models Q$.
\end{proof}

\subsection{Sufficient conditions for completeness}
Here we prove Th.\,\ref{th:completeness},  \ref{th:completeness:pruned},
and  Prop.\,\ref{prop:cssld.complete2}.
We begin with a slightly more general version of Th.\,\ref{prop:cssld.complete},
from which Th.\,\ref{prop:cssld.complete} follows immediately.
The only difference is condition~\ref{prop:cssld.complete.cond2a}.

\begin{lemma}
\label{lemma:cssld.complete}%
\renewcommand{\label}[1]{}%
\thcompletenesspruned
{  all the branches of $T$ are finite, or}
\end{lemma}

\begin{proof}
Let $Q'$ be a node of $T$ which has a
ground instance $Q'\sigma$ such that $S\models Q'\sigma$.  
Let an atom $A$ of $Q'$ and a program $\Pi_i$ be selected in $Q'$.
Now $A\sigma\in S$, $A\sigma\not\in S\setminus S_i$, thus 
$A\sigma\in S_i$, and $A\sigma$ is covered by $\Pi_i$ w.r.t.\ $S$.
Let $A\sigma\gets\seq[m]B$ be a ground instance of a clause from $\Pi_i$,
with $\seq[m]B\in S$.  Let $Q''$ be $Q'\sigma$ with $A\sigma$ replaced by 
$\seq[m]B$.  Then, by the lifting theorem  \cite[Th.\,5.37]{Doets},
$Q''$ is an instance of a child $Q^+$ of $Q'$ in $T$.
Obviously, $S\models Q''$.  

Assume the root of $T$ is $Q$.
Let $S\models Q\theta$ for a ground instance $Q\theta$ of $Q$.
By induction, 
using the previous paragraph as the inductive step,
there exists a branch $\Delta$ in $T$ which
(i)~is infinite or its last node is the empty query,
(ii) each its nonempty node has
a ground instance consisting of atoms from $S$,
(iii) the sequence of these ground instances is a
derivation $\Gamma$ for $Q\theta$ and $ground\left(P\right)$,
 and (iv) $\Delta$ is a lift \cite{Doets} of $\Gamma$.

We show that $\Delta$ is finite. 
This is immediately implied by condition 3a of the Lemma.  
Condition \ref{prop:cssld.complete.cond2b} or \ref{prop:cssld.complete.cond2c} 
implies that $ground(P)$ is recurrent or acceptable.
Hence $\Gamma$ is finite (by \cite[Corollaries 6.10, 6.25]{Apt-Prolog});
thus $\Delta$ is finite too.
By (i), $\Delta$ is a successful derivation (for $Q$ and $P$),
thus so is $\Gamma$, by (iv).
In the notation of \cite{Doets}, the respective answers of $\Gamma,\Delta$ are 
$Q\theta= res(\Gamma)$ and $res(\Delta)$
(the resultants of the derivations).
By the lifting theorem, $Q\theta$ is an instance of $res(\Delta)$,
which is an answer of $T$.
\end{proof}

Now, let us take a program $P$ and a specification $S$, and let
$n=1,\ \Pi_1=P,\linebreak[3]\ S_1=S$.
 $P$ is suitable for any atom, as $S\setminus S_1=\emptyset$.
Thus each SLD-tree $T$ for $P$ is a csSLD-tree compatible with $P,S$, and
from Lemma \ref{lemma:cssld.complete} it follows:

  \begin{proposition}
  \label{prop:compl:infi}
  \propcompletenessinfinite
  \end{proposition}

Th.\,\ref{th:completeness} follows immediately from
Prop.\,\ref{prop:compl:infi}.  
Now we show an example where Th.\,\ref{th:completeness} is not applicable,
but Prop.\,\ref{prop:compl:infi} is.

\begin{example}
\label{ex:prop:compl:infi}
Program
    $P = \{\, p(s(X))\mathop{\gets} p(X).\,\ \linebreak[3] p(0).\ \,
    q(X)\mathop\gets p(s(Y)).\}
    $
is complete w.r.t.\
    $S = \{ p(s^i(0))\mid i\geq0 \} \cup \{ q(0) \}$,
but its completeness does not follow from semi-completeness.
The program loops for any instance of $Q=q(X)$ (the SLD-tree has an infinite
branch).   However all the branches of
the SLD-tree for $q(X)$ and $ground(P)$ are finite
(although the tree has infinite height).
Hence by  Prop.\,\ref{prop:compl:infi}, $ground(P)$ is complete for $Q$.
Thus so is $P$, as both programs have the same least Herbrand model.

\end{example}

\medskip

{\sc Proposition \ref{prop:cssld.complete2}}.
\propcompletenesspruned

\begin{proof}
For a ground query $Q=\seq A$, let us define
the level of $Q$  as the multiset $|Q| = bag(|A_1|, \ldots, |A_n|)$,
whenever $|A_1|, \ldots, |A_n|$ are defined.
Let $\prec_m$ be the multiset ordering \cite{Apt-Prolog}.

We follow the proof of Lemma \ref{lemma:cssld.complete}.
The first part of the proof is basically the same, with an additional
restriction made possible by
the atom $A\sigma\in S$ being recurrently covered by $\Pi_i$ w.r.t.\ $S$.
Namely
for the ground clause $A\sigma\gets\seq[m]B$ we additionally require that
$|A\sigma|, |B_1|, \ldots |B_n|$ are defined, and
and $|A\sigma|>|B_i|$ for all $i=1,\ldots,n$.
Hence $|Q''|\prec_m |Q'\sigma|$.

As in the proof of Lemma \ref{lemma:cssld.complete},
there is a branch $\Delta$ in $T$ and a derivation $\Gamma$ for $Q\theta$ and
$ground(P)$, satisfying conditions (i),\ldots,(iv).  For each query $Q''$ of
$\Gamma$, $S\models Q''$, and $|Q''|$ is defined.
As the levels of the queries in $\Gamma$ are a $\prec_m$-decreasing sequence,
and $\prec_m$ is well founded, $\Gamma$ is finite.  
Hence  $\Gamma$ and $\Delta$ are successful,
and $Q\theta$ is an instance of an answer of~$T$.
\end{proof}

\bibliographystyle{ACM-Reference-Format-Journals}
\pdfbookmark[1]{References}{bookmark:bib}
{%
 \renewcommand{\small}{}
 \bibliography{bibshorter,bibmagic,bibpearl}
}

\end{document}